\documentclass[a4paper,onecolumn,11pt,unpublished]{quantumarticle}
\pdfoutput=1

\usepackage[utf8]{inputenc}
\usepackage[english]{babel}
\usepackage[T1]{fontenc}
\usepackage[numbers]{natbib}
\usepackage{amsmath,amssymb,amsfonts,amsthm}
\usepackage{mathtools}
\usepackage{graphicx}
\usepackage{xcolor}
\usepackage{tikz}
\usepackage{pgfplots}
\usepackage{tikz-cd}
\usepackage{braket}
\usepackage{array}
\usepackage{booktabs}
\usepackage{enumitem}
\usepackage{longtable}
\usepackage{nicefrac}
\usepackage{stmaryrd}
\pdfmapfile{+stmaryrd.map}
\pdfmapfile{+rsfs.map}
\usepackage{hyperref}
\usepackage[nameinlink]{cleveref}

\usetikzlibrary{trees}
\usetikzlibrary{topaths}
\usetikzlibrary{decorations.pathmorphing}
\usetikzlibrary{decorations.markings}
\usetikzlibrary{matrix,backgrounds,folding}
\usetikzlibrary{chains,scopes,positioning,fit}
\usetikzlibrary{arrows,shadows}
\usetikzlibrary{calc}
\usetikzlibrary{shapes,shapes.geometric,shapes.misc}
\usetikzlibrary{decorations.pathreplacing}

\pgfdeclarelayer{edgelayer}
\pgfdeclarelayer{nodelayer}
\pgfsetlayers{background,edgelayer,nodelayer,main}
\pgfplotsset{compat=1.18}
\tikzstyle{none}=[inner sep=0mm]
\tikzstyle{every loop}=[]
\tikzstyle{env}=[copoint,regular polygon rotate=0,minimum width=0.2cm, fill=black]

\tikzstyle{hadamard edge}=[-,color=blue,dashed,dash pattern=on 2pt off 0.7pt]
\tikzstyle{black edge}=[-,color=black,dashed,dash pattern=on 2pt off 0.7pt]

\tikzstyle{every picture}=[baseline=-0.25em]
\tikzstyle{dotpic}=[scale=0.5]
\tikzstyle{diredges}=[every to/.style={diredge}]
\tikzstyle{dot graph}=[shorten <=-0.1mm,shorten >=-0.1mm,scale=0.6]
\tikzstyle{plot point}=[circle,fill=black,minimum width=2mm,inner sep=0]

\tikzstyle{braceedge}=[decorate,decoration={brace,amplitude=2mm,raise=-1mm}]
\tikzstyle{small braceedge}=[decorate,decoration={brace,amplitude=1mm,raise=-1mm}]
\tikzstyle{left hook arrow}=[left hook-latex]
\tikzstyle{right hook arrow}=[right hook-latex]

\tikzstyle{black dot}=[inner sep=0.7mm,minimum width=0pt,minimum height=0pt,fill=black,draw=black,shape=circle]

\tikzstyle{dot}=[black dot]
\tikzstyle{smalldot}=[inner sep=0.4mm,minimum width=0pt,minimum height=0pt,fill=black,draw=black,shape=circle]%
\tikzstyle{white dot}=[dot,fill=white]
\tikzstyle{antipode}=[white dot,inner sep=0.3mm,font=\footnotesize]
\tikzstyle{smallwhitedot}=[smalldot,fill=white]%
\tikzstyle{alt white dot}=[white dot,label={[xshift=3.07mm,yshift=-0.05mm,font=\footnotesize]left:$*$}]
\tikzstyle{gray dot}=[dot,fill=gray!40!white]
\tikzstyle{smallgraydot}=[smalldot,fill=gray!40!white]%
\tikzstyle{box vertex}=[draw=black,rectangle]
\tikzstyle{small box}=[box vertex,fill=white]%
\tikzstyle{whitebg}=[fill=white,inner sep=2pt]
\tikzstyle{graph state vertex}=[sg vertex,fill=black]

\tikzstyle{wide copoint}=[fill=white,draw=black,shape=isosceles triangle,shape border rotate=90,isosceles triangle stretches=true,inner sep=1pt,minimum width=1.5cm,minimum height=5mm]
\tikzstyle{wide point}=[fill=white,draw=black,shape=isosceles triangle,shape border rotate=-90,isosceles triangle stretches=true,inner sep=1pt,minimum width=1.5cm,minimum height=4mm]
\tikzstyle{very wide copoint}=[fill=white,draw=black,shape=isosceles triangle,shape border rotate=-90,isosceles triangle stretches=true,inner sep=1pt,minimum width=2.5cm,minimum height=4mm]
\tikzstyle{very wide empty copoint}=[draw=black,shape=isosceles triangle,shape border rotate=-90,isosceles triangle stretches=true,inner sep=1pt,minimum width=2.5cm,minimum height=4mm]
\tikzstyle{symm}=[ultra thick,shorten <=-1mm,shorten >=-1mm]

\tikzstyle{square box}=[rectangle,fill=white,draw=black,minimum height=5mm,minimum width=5mm,font=\small]
\tikzstyle{square gray box}=[rectangle,fill=gray!30,draw=black,minimum height=6mm,minimum width=6mm]
\tikzstyle{wide medium box}=[shape=rectangle, text height=1.5ex, text depth=0.25ex, yshift=0.5mm, fill=white, draw=black, minimum height=10mm, yshift=-0.5mm, minimum width=7.5mm, font={\small}]
\tikzstyle{copoint}=[regular polygon,regular polygon sides=3,draw=black,scale=0.75,inner sep=-0.5pt,minimum width=7mm,fill=white]
\tikzstyle{point}=[regular polygon,regular polygon sides=3,draw=black,scale=0.75,inner sep=-0.5pt,minimum width=7mm,fill=white,regular polygon rotate=180]
\tikzstyle{gray point}=[point,fill=gray!40!white]
\tikzstyle{gray copoint}=[copoint,fill=gray!40!white]
\tikzstyle{pointer edge}=[->, very thick, gray]

\newcommand{\edgearrow}{{\arrow[black]{>}}}
\newcommand{\edgetick}{{\arrow[black,scale=0.7,very thick]{|}}}

\tikzstyle{diredge}=[->]
\tikzstyle{rdiredge}=[<-]
\tikzstyle{medium diredge}=[->]

\tikzstyle{short diredge}=[->]
\tikzstyle{halfedge}=[-)]
\tikzstyle{other halfedge}=[(-]
\tikzstyle{freeedge}=[(-)]
\tikzstyle{white edge}=[line width=5pt,white]
\tikzstyle{tick}=[postaction=decorate,decoration={markings, mark=at position 0.5 with \edgetick}]
\tikzstyle{small map edge}=[|-latex, gray!60!blue, shorten <=0.9mm, shorten >=0.5mm]
\tikzstyle{thick dashed edge}=[very thick,dashed,gray!40]
\tikzstyle{map edge}=[|-latex,very thick, gray!40, shorten <=1mm, shorten >=0.5mm]
\tikzstyle{tickedge}=[postaction=decorate,
  decoration={markings, mark=at position 0.5 with \edgetick}]
\tikzstyle{dirtickedge}=[postaction=decorate,
  decoration={markings, mark=at position 0.5 with \edgetick},
  decoration={markings, mark=at position 0.85 with \edgearrow}]
\tikzstyle{dirdoubletickedge}=[postaction=decorate,
  decoration={markings, mark=at position 0.4 with \edgetick},
  decoration={markings, mark=at position 0.6 with \edgetick},
  decoration={markings, mark=at position 0.85 with \edgearrow}]

\makeatletter
\newcommand{\boxshape}[3]{%
\pgfdeclareshape{#1}{
\inheritsavedanchors[from=rectangle]%
\inheritanchorborder[from=rectangle]
\inheritanchor[from=rectangle]{center}
\inheritanchor[from=rectangle]{north}
\inheritanchor[from=rectangle]{south}
\inheritanchor[from=rectangle]{west}
\inheritanchor[from=rectangle]{east}
\backgroundpath{%
\southwest \pgf@xa=\pgf@x \pgf@ya=\pgf@y
\northeast \pgf@xb=\pgf@x \pgf@yb=\pgf@y

\@tempdima=#2
\@tempdimb=#3

\pgfpathmoveto{\pgfpoint{\pgf@xa - 5pt + \@tempdima}{\pgf@ya}}
\pgfpathlineto{\pgfpoint{\pgf@xa - 5pt - \@tempdima}{\pgf@yb}}
\pgfpathlineto{\pgfpoint{\pgf@xb + 5pt + \@tempdimb}{\pgf@yb}}
\pgfpathlineto{\pgfpoint{\pgf@xb + 5pt - \@tempdimb}{\pgf@ya}}
\pgfpathlineto{\pgfpoint{\pgf@xa - 5pt + \@tempdima}{\pgf@ya}}
\pgfpathclose
}
}}

\boxshape{NEbox}{0pt}{8pt}
\boxshape{SEbox}{0pt}{-8pt}
\boxshape{NWbox}{8pt}{0pt}
\boxshape{SWbox}{-8pt}{0pt}
\makeatother

\tikzstyle{map}=[draw,shape=NEbox,inner sep=7pt]
\tikzstyle{mapdag}=[draw,shape=SEbox,inner sep=7pt]
\tikzstyle{maptrans}=[draw,shape=SWbox,inner sep=7pt]
\tikzstyle{mapconj}=[draw,shape=NWbox,inner sep=7pt]

\tikzstyle{probs}=[shape=semicircle,fill=gray!40!white,draw=black,shape border rotate=180,minimum width=1.2cm]

\tikzstyle{arrs}=[-latex,font=\small,auto]
\tikzstyle{arrow plain}=[arrs]
\tikzstyle{arrow dashed}=[dashed,arrs]
\tikzstyle{arrow bold}=[very thick,arrs]
\tikzstyle{arrow hide}=[draw=white!0,-]
\tikzstyle{arrow reverse}=[latex-]
\tikzstyle{cdnode}=[]

\tikzstyle{gn}=[dot,fill=lime!50,minimum width=0.2cm,inner sep=0.5pt,font=\footnotesize]
\tikzstyle{rn}=[dot,fill=red!50,inner sep=0.5pt,minimum width=0.2cm,font=\footnotesize]
\tikzstyle{bn}=[dot,fill=blue,minimum width=0.3cm]

\tikzstyle{rc}=[dot,thick,fill=white,draw = red,minimum width=0.2cm,inner sep=0.5pt,font=\footnotesize]
\tikzstyle{gc}=[dot,thick,fill=white,draw= lime,inner sep=0.5pt,minimum width=0.2cm,font=\footnotesize]
\tikzstyle{bc}=[dot,thick,fill=white,draw= blue,minimum width=0.3cm]

\tikzstyle{label}=[circle,fill=white,minimum width=0.3cm]

\tikzstyle{H box}=[rectangle,fill=yellow,draw=black,xscale=1,yscale=1,font=\small,inner sep=0.75pt,minimum width=0.15cm,minimum height=0.15cm]

\tikzstyle{clocklabel}=[dot,fill=yellow,draw=black,font=\tiny,inner sep=0.75pt]

\tikzstyle{rsn}=[circle split,draw,fill=red,font=\tiny,inner sep=0.75pt]
\tikzstyle{gsn}=[circle split,draw,fill=lime,font=\tiny,inner sep=0.75pt]
\tikzstyle{bsn}=[circle split,draw,fill=blue,font=\tiny,inner sep=0.75pt]

\tikzstyle{rsc}=[circle split,thick,draw= red,draw,fill=white,font=\tiny,inner sep=0.75pt]
\tikzstyle{gsc}=[circle split,thick,draw= lime,draw,fill=white,font=\tiny,inner sep=0.75pt]
\tikzstyle{bsc}=[circle split,thick,draw= blue,draw,fill=white,font=\tiny,inner sep=0.75pt]

\tikzstyle{cnot}=[fill=white,shape=circle,inner sep=-1.4pt]
\tikzstyle{wire label}=[font=\tiny, auto]

\tikzstyle{Z dot}=[inner sep=0mm,minimum size=2mm,shape=circle,draw=black,
  fill={rgb,255:red,223;green,255;blue,128},line width=.4pt]
\tikzstyle{X dot}=[Z dot,
  fill={rgb,255:red,255;green,128;blue,128}]
\tikzstyle{Z phase dot}=[inner sep=0mm,minimum size=5.47mm,shape=circle,
  draw=black,fill={rgb,255:red,223;green,255;blue,128},line width=.4pt]
\tikzstyle{X phase dot}=[Z phase dot,
  fill={rgb,255:red,255;green,128;blue,128}]
\tikzstyle{hadamard}=[fill=yellow,draw=black,shape=rectangle,inner sep=0.6mm,
  minimum height=1.5mm,minimum width=1.5mm,line width=.4pt]
\tikzstyle{empty diagram}=[draw={rgb,255:red,170;green,170;blue,170},dashed,
  shape=rectangle,inner sep=0mm,minimum height=5mm,minimum width=5mm,line width=.4pt]

\tikzstyle{tikzfig}=[baseline=-0.25em,scale=0.5]

\newcommand{\widetikzitinput}[1]{%
	\begin{center}
	\makebox[\linewidth][c]{%
	\begingroup
	\scalebox{0.8}{{\tikzstyle{every picture}=[tikzfig]\input{#1}}}%
	\endgroup
}%
	\end{center}
}

\newcommand{\intf}[1]{\ensuremath{\left\llbracket #1 \right\rrbracket}}
\newcommand{\F}{\mathbb F}
\newcommand{\Z}{\mathbb Z}
\DeclareRobustCommand{\ZXsimp}{\ensuremath{\text{\bfseries\upshape ZX}_{\text{\sffamily\bfseries\upshape simp}}}}
\newcommand{\id}{\operatorname{id}}
\newcommand{\Hom}{\operatorname{Hom}}
\newcommand{\diag}{\operatorname{diag}}
\newcommand{\daggerfootnote}[1]{%
	\begingroup
	\renewcommand{\thefootnote}{\ensuremath{\dagger}}%
	\footnote{#1}%
	\endgroup
}
\newcommand{\ddaggerfootnote}[1]{%
	\begingroup
	\renewcommand{\thefootnote}{\ensuremath{\ddagger}}%
	\footnote{#1}%
	\endgroup
}

\newcommand{\semzx}[1]{%
  \left\llbracket\vbox{\hbox{%
\begingroup%
\tikzstyle{every picture}=[tikzfig]%
\input{diagrams/#1.tikz}%
\endgroup%
}}\right\rrbracket^{(B2')}%
}
\newcommand{\ocross}{\otimes}
\newcommand{\sOneLhsDiagram}{\vcenter{\hbox{%
\begin{tikzpicture}[font={\footnotesize}]
	\begin{pgfonlayer}{nodelayer}
		\node [style=gn] (3) at (-0.75, -0.25) {\footnotesize$~\ell~$};
		\node [style=none] (4) at (-1.75, -0.5) {\raisebox{2mm}{...}};
		\node [style=none] (6) at (-1, -0.75) {};
		\node [style=none] (8) at (-0.5, -0.75) {};
		\node [style=none] (10) at (-2, -0.5) {};
		\node [style=none] (11) at (-1.5, -0.5) {};
		\node [style=none] (12) at (-0.75, -0.75) {\raisebox{2mm}{...}};
		\node [style=none] (15) at (-2, 0.75) {};
		\node [style=none] (16) at (-0.5, 0.5) {};
		\node [style=none] (17) at (-1.5, 0.75) {};
		\node [style=gn] (19) at (-1.75, 0.25) {\footnotesize$~k~$};
		\node [style=none] (20) at (-0.75, 0.5) {\raisebox{-2mm}{...}};
		\node [style=none] (21) at (-1, 0.5) {};
		\node [style=none] (22) at (-1.75, 0.75) {\raisebox{-2mm}{...}};
	\end{pgfonlayer}
	\begin{pgfonlayer}{edgelayer}
		\draw[bend right=23] (3) to (16.center);
		\draw[bend right=23] (3) to (6.center);
		\draw[bend left=23] (3) to (8.center);
		\draw[bend right=23] (19) to (10.center);
		\draw[bend left=23] (19) to (11.center);
		\draw (19) to (3);
		\draw[bend left=23]  (19) to (15.center);
		\draw[bend right=23]  (19) to (17.center);
		\draw[bend left=23] (3) to (21.center);
	\end{pgfonlayer}
\end{tikzpicture}}}}
\newcommand{\sOneRhsDiagram}{\vcenter{\hbox{%
\begin{tikzpicture}[font={\footnotesize}]
	\begin{pgfonlayer}{nodelayer}
		\node [style=gn] (z) at (0, 0) {\footnotesize$k+\ell$};
		\node [style=none] (t1) at (-0.5, 0.75) {};
		\node [style=none] (t2) at (0.5, 0.75) {};
		\node [style=none] (td) at (0, 0.75) {\raisebox{-2mm}{...}};
		\node [style=none] (b1) at (-0.5, -0.75) {};
		\node [style=none] (b2) at (0.5, -0.75) {};
		\node [style=none] (bd) at (0, -0.75) {\raisebox{2mm}{...}};
	\end{pgfonlayer}
	\begin{pgfonlayer}{edgelayer}
		\draw[bend left=23] (z) to (t1.center);
		\draw[bend right=23] (z) to (t2.center);
		\draw[bend right=23] (z) to (b1.center);
		\draw[bend left=23] (z) to (b2.center);
	\end{pgfonlayer}
\end{tikzpicture}}}}
\newcommand{\cupDiagram}{\vcenter{\hbox{%
\begin{tikzpicture}
	\begin{pgfonlayer}{nodelayer}
		\node [style=none] (l) at (0, -0.25) {};
		\node [style=none] (r) at (1, -0.25) {};
	\end{pgfonlayer}
	\begin{pgfonlayer}{edgelayer}
		\draw [in=90, out=90, looseness=1.75] (l.center) to (r.center);
	\end{pgfonlayer}
\end{tikzpicture}}}}
\newcommand{\greenCupDiagram}{\vcenter{\hbox{%
\begin{tikzpicture}
	\begin{pgfonlayer}{nodelayer}
		\node [style=none] (l) at (0, -0.25) {};
		\node [style=none] (r) at (1, -0.25) {};
		\node [style=gn] (g) at (0.5, 0.25) {};
	\end{pgfonlayer}
	\begin{pgfonlayer}{edgelayer}
		\draw [in=90, out=90, looseness=1.75] (l.center) to (r.center);
	\end{pgfonlayer}
\end{tikzpicture}}}}
\newcommand{\redCupDiagram}{\vcenter{\hbox{%
\begin{tikzpicture}
	\begin{pgfonlayer}{nodelayer}
		\node [style=none] (l) at (0, -0.25) {};
		\node [style=none] (r) at (1, -0.25) {};
		\node [style=rn] (x) at (0.5, 0.25) {};
	\end{pgfonlayer}
	\begin{pgfonlayer}{edgelayer}
		\draw [in=90, out=90, looseness=1.75] (l.center) to (r.center);
	\end{pgfonlayer}
\end{tikzpicture}}}}
\newcommand{\bOneLhsDiagram}{\vcenter{\hbox{%
\begin{tikzpicture}
	\begin{pgfonlayer}{nodelayer}
		\node [style=rn] (r1) at (0, 0.35) {};
		\node [style=gn] (g1) at (0, -0.25) {};
		\node [style=rn] (r2) at (0.85, 0.35) {};
		\node [style=gn] (g2) at (0.85, -0.25) {};
		\node [style=none] (a) at (0.6, -0.8) {};
		\node [style=none] (b) at (1.1, -0.8) {};
	\end{pgfonlayer}
	\begin{pgfonlayer}{edgelayer}
		\draw (r1) to (g1);
		\draw (r2) to (g2);
		\draw[bend right=23] (g2) to (a.center);
		\draw[bend left=23] (g2) to (b.center);
	\end{pgfonlayer}
\end{tikzpicture}}}}
\newcommand{\bOneRhsDiagram}{\vcenter{\hbox{%
\begin{tikzpicture}
	\begin{pgfonlayer}{nodelayer}
		\node [style=rn] (r1) at (0, 0.25) {};
		\node [style=rn] (r2) at (0.55, 0.25) {};
		\node [style=none] (b1) at (0, -0.45) {};
		\node [style=none] (b2) at (0.55, -0.45) {};
	\end{pgfonlayer}
	\begin{pgfonlayer}{edgelayer}
		\draw (r1) to (b1.center);
		\draw (r2) to (b2.center);
	\end{pgfonlayer}
\end{tikzpicture}}}}
\newcommand{\zOLhsDiagram}{\vcenter{\hbox{%
\begin{tikzpicture}
	\begin{pgfonlayer}{nodelayer}
		\node [style=gn] (pi) at (0, 0) {$\pi$};
		\node [style=gn] (g) at (0.7, 0.25) {};
		\node [style=none] (out) at (0.7, -0.25) {};
	\end{pgfonlayer}
	\begin{pgfonlayer}{edgelayer}
		\draw (g) to (out.center);
	\end{pgfonlayer}
\end{tikzpicture}}}}
\newcommand{\zORhsDiagram}{\vcenter{\hbox{%
\begin{tikzpicture}
	\begin{pgfonlayer}{nodelayer}
		\node [style=gn] (pi) at (0, 0) {$\pi$};
		\node [style=rn] (r) at (0.7, 0.25) {};
		\node [style=none] (out) at (0.7, -0.25) {};
	\end{pgfonlayer}
	\begin{pgfonlayer}{edgelayer}
		\draw (r) to (out.center);
	\end{pgfonlayer}
\end{tikzpicture}}}}
\newcommand{\hLhsDiagram}{\vcenter{\hbox{%
\begin{tikzpicture}
	\begin{pgfonlayer}{nodelayer}
		\node [style=rn] (x) at (0, 0) {\footnotesize$k$};
		\node [style={H box}] (h1) at (-0.45, 0.6) {};
		\node [style={H box}] (h2) at (0.45, 0.6) {};
		\node [style={H box}] (h3) at (-0.45, -0.6) {};
		\node [style={H box}] (h4) at (0.45, -0.6) {};
		\node [style=none] (t1) at (-0.45, 1.0) {};
		\node [style=none] (t2) at (0.45, 1.0) {};
		\node [style=none] (b1) at (-0.45, -1.0) {};
		\node [style=none] (b2) at (0.45, -1.0) {};
		\node [style=none] (td) at (0, 1.0) {\raisebox{-2mm}{...}};
		\node [style=none] (bd) at (0, -1.0) {\raisebox{2mm}{...}};
	\end{pgfonlayer}
	\begin{pgfonlayer}{edgelayer}
		\draw[bend left=23] (x) to (h1);
		\draw[bend right=23] (x) to (h2);
		\draw[bend right=23] (x) to (h3);
		\draw[bend left=23] (x) to (h4);
		\draw (h1) to (t1.center);
		\draw (h2) to (t2.center);
		\draw (h3) to (b1.center);
		\draw (h4) to (b2.center);
	\end{pgfonlayer}
\end{tikzpicture}}}}
\newcommand{\hRhsDiagram}{\vcenter{\hbox{%
\begin{tikzpicture}
	\begin{pgfonlayer}{nodelayer}
		\node [style=gn] (z) at (0, 0) {\footnotesize$k$};
		\node [style=none] (t1) at (-0.45, 0.75) {};
		\node [style=none] (t2) at (0.45, 0.75) {};
		\node [style=none] (b1) at (-0.45, -0.75) {};
		\node [style=none] (b2) at (0.45, -0.75) {};
		\node [style=none] (td) at (0, 0.75) {\raisebox{-2mm}{...}};
		\node [style=none] (bd) at (0, -0.75) {\raisebox{2mm}{...}};
	\end{pgfonlayer}
	\begin{pgfonlayer}{edgelayer}
		\draw[bend left=23] (z) to (t1.center);
		\draw[bend right=23] (z) to (t2.center);
		\draw[bend right=23] (z) to (b1.center);
		\draw[bend left=23] (z) to (b2.center);
	\end{pgfonlayer}
\end{tikzpicture}}}}
\newcommand{\ivLhsDiagram}{\vcenter{\hbox{%
\begin{tikzpicture}
	\begin{pgfonlayer}{nodelayer}
		\node [style=gn] (s) at (-0.75, 0) {};
		\node [style=rn] (r1) at (-0.25, 0.45) {};
		\node [style=gn] (g1) at (-0.25, -0.35) {};
		\node [style=rn] (r2) at (0.35, 0.45) {};
		\node [style=gn] (g2) at (0.35, -0.35) {};
	\end{pgfonlayer}
	\begin{pgfonlayer}{edgelayer}
		\draw (r1) to (g1);
		\draw[bend left=38] (r1) to (g1);
		\draw[bend right=38] (r1) to (g1);
		\draw (r2) to (g2);
		\draw[bend left=38] (r2) to (g2);
		\draw[bend right=38] (r2) to (g2);
	\end{pgfonlayer}
\end{tikzpicture}}}}
\newcommand{\ivRhsDiagram}{\vcenter{\hbox{\begin{tikzpicture}
	\begin{pgfonlayer}{nodelayer}
		\node [style=none] (4) at (0.25, 0.25) {};
		\node [style=none] (5) at (-0.25, 0.25) {};
		\node [style=none] (6) at (0.25, -0.25) {};
		\node [style=none] (7) at (-0.25, -0.25) {};
	\end{pgfonlayer}
	\begin{pgfonlayer}{edgelayer}
		\draw [dashed, color=gray] (5.center) to (7.center);
		\draw [dashed, color=gray] (7.center) to (6.center);
		\draw [dashed, color=gray] (6.center) to (4.center);
		\draw [dashed, color=gray] (4.center) to (5.center);
	\end{pgfonlayer}
\end{tikzpicture}}}}
\newcommand{\euLhsDiagram}{\vcenter{\hbox{\begin{tikzpicture}
	\begin{pgfonlayer}{nodelayer}
		\node [style={H box}] (0) at (0.5, 0) {};
		\node [style=none] (1) at (0.5, 0.3) {};
		\node [style=none] (2) at (0.5, -0.3) {};
				\node [style=none] (a) at (0.5, 0.5) {};
		\node [style=none] (b) at (0.5, -0.5) {};
	\end{pgfonlayer}
	\begin{pgfonlayer}{edgelayer}
		\draw (1.center) to (0);
		\draw (2.center) to (0);
	\end{pgfonlayer}
\end{tikzpicture}}}}
\newcommand{\euRhsDiagram}{\vcenter{\hbox{%
\begin{tikzpicture}
	\begin{pgfonlayer}{nodelayer}
		\node [style=none] (top) at (0, 1.0) {};
		\node [style=gn] (g1) at (0, 0.55) {\footnotesize$1$};
		\node [style=rn] (x) at (0, 0.1) {};
		\node [style=gn] (g2) at (0, -0.35) {\footnotesize$1$};
		\node [style=none] (bot) at (0, -0.85) {};
		\node [style=gn] (g3) at (0.55, 0.35) {\footnotesize$-1$};
	\end{pgfonlayer}
	\begin{pgfonlayer}{edgelayer}
		\draw (top.center) to (g1);
		\draw (g1) to (x);
		\draw (x) to (g2);
		\draw (g2) to (bot.center);
		\draw (g3) to (x);
	\end{pgfonlayer}
\end{tikzpicture}}}}
\newcommand{\bTwoLhsDiagram}{\vcenter{\hbox{%
\begin{tikzpicture}
	\begin{pgfonlayer}{nodelayer}
		\node [style=rn] (1) at (0.5, -0.25) {};
		\node [style=none] (3) at (1.25, 1) {};
		\node [style=none] (5) at (0.5, -0.75) {};
		\node [style=none] (6) at (0.5, 1) {};
		\node [style=gn] (7) at (1.25, 0.5) {};
		\node [style=rn] (9) at (1.25, -0.25) {};
		\node [style=gn] (11) at (0.5, 0.5) {};
		\node [style=none] (12) at (1.25, -0.75) {};
		\node [style=gn] (17) at (0, 0) {};
	\end{pgfonlayer}
	\begin{pgfonlayer}{edgelayer}
		\draw [style=none] (12.center) to (9);
		\draw [style=none] (5.center) to (1);
		\draw [style=none] (7) to (3.center);
		\draw [style=none, bend right=23, looseness=1.00] (9) to (7);
		\draw [style=none] (11) to (6.center);
		\draw [style=none, bend left=23, looseness=1.00] (1) to (11);
		\draw (11) to (9);
		\draw (7) to (1);
	\end{pgfonlayer}
\end{tikzpicture}}}}
\newcommand{\bTwoRhsDiagram}{\vcenter{\hbox{%
\begin{tikzpicture}
	\begin{pgfonlayer}{nodelayer}
		\node [style=none] (0) at (0.75, 0.75) {};
		\node [style=none] (2) at (0.25, -0.75) {};
		\node [style=none] (4) at (0.25, 0.75) {};
		\node [style=gn] (10) at (0.5, -0.25) {};
		\node [style=none] (13) at (0.75, -0.75) {};
		\node [style=rn] (14) at (0.5, 0.25) {};
		\node [style=rn] (15) at (0, 0.25) {};
		\node [style=gn] (16) at (0, -0.25) {};
	\end{pgfonlayer}
	\begin{pgfonlayer}{edgelayer}
		\draw [style=none, bend right=23, looseness=1.00] (13.center) to (10);
		\draw [style=none] (10) to (14);
		\draw [style=none, bend left=23, looseness=1.00] (14) to (4.center);
		\draw [style=none, bend right=23, looseness=1.00] (14) to (0.center);
		\draw [bend right=23, looseness=1.00] (10) to (2.center);
		\draw (15) to (16);
	\end{pgfonlayer}
\end{tikzpicture}}}}

\newtheorem{theorem}{Theorem}[section]
\newtheorem{lemma}[theorem]{Lemma}
\newtheorem{corollary}[theorem]{Corollary}
\theoremstyle{definition}
\newtheorem{definition}[theorem]{Definition}

\title{Minimality of the stabilizer ZX calculus}
\author{Harry K. Stoltz}
\orcid{0009-0007-3805-5276}
\email{harry.stoltz@nyu.edu}
\affiliation{Courant Institute of Mathematical Sciences, New York University}

\begin{document}

\maketitle

\begin{abstract}
The stabilizer fragment of the ZX calculus is amongst the most important fragments of the theory.  Crucially, the stabilizer calculus can be described by a small collection of rewrites, most of which have been shown to be necessary by Backens--Perdrix--Wang \cite{backens2020towards}.  However, two rules, describing the red/green compact-structure coincidence and the important bialgebra law, had not been shown to be necessary.  We present two alternate interpretations, showing that both of these rules are individually necessary.  Thus, this resolves a problem which has remained open for over a decade, and gives the first complete, minimal ruleset for the stabilizer ZX Calculus.
\end{abstract}

\section*{Summary}

Mathematical calculi for quantum circuits often contain many rewrite rules, and a natural question is whether any of those rules are redundant. This paper studies a small complete rule set for the stabilizer fragment of the ZX calculus and proves that the two rules whose necessity was previously unresolved are each independent from the rest, assuming the usual connectivity-only meta-rule. Crucially, the bialgebra law is proven to be a necessary rewrite rule.

\section{Introduction}
\label{sec:introduction}

The ZX calculus is a graphical language for reasoning about quantum computation \cite{coecke2011interacting}, arising out of the study of Categorical Quantum Mechanics (CQM) \cite{abramsky2004categorical}.  The stabilizer formalism \cite{gottesman1997stabilizercodesquantumerror} is central in quantum error correction \cite{nielsen2010quantum} and measurement-based quantum computation \cite{raussendorf2001one}, and the ZX calculus has also recently been applied to photonic quantum computing \cite{de2022quantum}.  The ZX calculus and CQM have also found their place as a means of communicating quantum mechanical ideas to a broader audience, without the need for matrix formalism \cite{dundar2023quantum}.  We refer the interested reader to the following survey by van de Wetering for an overview of the ZX calculus \cite{van2020zx}, and the books by Coecke--Kissinger \cite{coecke2018picturing}, or Heunen--Vicary \cite{heunen2019categories} for a categorical introduction to CQM.

The stabilizer fragment of the ZX calculus is an important cornerstone of the broader theory, both historically and practically. The stabilizer fragment restricts spider phases to integer multiples of \(\pi/2\), and was the first version of the calculus to be found complete \cite{backens2014zx}. Wider completeness results have since been established \cite{ng2017universal}, and the closely related Clifford+\(T\) fragment, with phases in integer multiples of \(\pi/4\), is approximately universal. However, the stabilizer fragment remains an important object of study.  The rule set for the stabilizer ZX calculus also captures, in some sense, some of the most fundamental quantum properties among its rules.

Backens et al. \cite{backens2020towards} established a significantly smaller complete rule set for the stabilizer ZX calculus with just nine rules (Table 1) and one meta-rule “only connectivity matters.”  Previous rule sets usually contained more explicit rules and conventions allowing for flipping of diagrams or switching colors, whereas this rule set allowed for their derivation instead.  Backens et al. proved that seven of the nine rules were necessary (that is, they could not be derived from the other rules) for the overall rule set to be complete, and amongst the remaining two rules, (S3’R) and (B2’), at least one was necessary.  The rule (S3’R) essentially says that the ‘red’ compact structure coincides with the ‘green’ one (which coincides with the ambient compact structure of the category by (S3’L)).  It’s interesting that proving independence of the (B2’) bialgebra rule was so elusive, since this rule describes an important feature of the ZX calculus: complementarity.  Notably absent from the \(\ZXsimp\) rule set is the Hopf rule, which lets one ‘chop off’ double wires connecting red and green spiders.  This is  usually derived from the bialgebra law (or vice versa) so it would be surprising to be able to derive either of these rules from the remaining axioms.

In this paper, we prove that both of the remaining rules, (S3’R) and (B2’) are, in fact, individually necessary relative to the connectivity meta-rule of Backens--Perdrix--Wang \cite{backens2020towards}.  This means that \(\ZXsimp\) + the meta-rule has no redundant rewrite rule.  This also reinforces that the complementarity property encoded by the bialgebra law, and the compact structure induced by the ‘red’ spider are distinct properties from each other, as well as the rest of the rule set.

\section{The Stabilizer ZX Calculus}
\label{sec:stabilizer-zx-calculus}

We now describe the stabilizer ZX calculus in more depth, and outline some of the conventions we use.  In this paper, we mirror the setup of Backens et al. \cite{backens2020towards} for continuity.

In the stabilizer ZX calculus, a diagram \(D\) with \(k\) inputs and \(l\) outputs  \(D:k\to l\) is composed of the following generators:

\begingroup
\footnotesize
\setlength{\tabcolsep}{4pt}
\renewcommand{\arraystretch}{1.35}
\begin{center}
\begin{tabular}{|>{\centering\arraybackslash}p{0.23\textwidth}|>{\centering\arraybackslash}p{0.33\textwidth}|>{\centering\arraybackslash}p{0.34\textwidth}|}
\hline
Name & Diagram & Typed\\
\hline
Green Spider & \(\vcenter{\hbox{\begin{tikzpicture}\footnotesize
	\begin{pgfonlayer}{nodelayer}
		\node [style=none] (0) at (-0.5, 0.5) {};
		\node [style=none] (1) at (-1, 0.5) {};
		\node [style=none] (2) at (0, 0.45) {...};
		\node [style=none] (3) at (0, -0.45) {...};
		\node [style=none] (4) at (0.5, -0.5) {};
		\node [style=none] (5) at (0.75, 0.75) {};
		\node [style=none] (6) at (0.75, -0.75) {};
		\node [style=gn] (7) at (-0.25, -0) {$~\alpha~$};
		\node [style=none] (8) at (-1.25, 0.75) {};
		\node [style=none] (9) at (-1.25, -0.75) {};
		\node [style=none] (10) at (-1, -0.5) {};
		\node [style=none] (11) at (0.5, 0.5) {};
		\node [style=none] (12) at (-0.5, -0.5) {};
	\end{pgfonlayer}
	\begin{pgfonlayer}{edgelayer}
		\draw [style=braceedge] (6.center) to node[wire label, inner sep={5 pt}]{$m$} (9.center);
		\draw [style=none, bend left=15, looseness=1.00] (10.center) to (7);
		\draw [style=none, bend left=15, looseness=1.00] (12.center) to (7);
		\draw [style=braceedge] (8.center) to node[wire label, inner sep={5 pt}]{$n$} (5.center);
		\draw [style=none, bend left=15, looseness=1.00] (7) to (1.center);
		\draw [style=none, bend right=15, looseness=1.00] (4.center) to (7);
		\draw [style=none, bend right=15, looseness=1.00] (7) to (11.center);
		\draw [style=none, bend left=15, looseness=1.00] (7) to (0.center);
	\end{pgfonlayer}
\end{tikzpicture}}}\) & \(Z^\alpha_{n,m}:n\to m\)\\[0.5em]
\hline
Red Spider & \(\vcenter{\hbox{\begin{tikzpicture}\footnotesize
	\begin{pgfonlayer}{nodelayer}
		\node [style=none] (0) at (1.25, -0.5) {};
		\node [style=none] (1) at (2.25, 0.5) {};
		\node [style=none] (2) at (0.5, -0.75) {};
		\node [style=none] (3) at (0.75, 0.5) {};
		\node [style=none] (4) at (0.5, 0.75) {};
		\node [style=rn] (5) at (1.5, 0) {$~\alpha~$};
		\node [style=none] (6) at (2.5, -0.75) {};
		\node [style=none] (7) at (1.75, -0.45) {...};
		\node [style=none] (8) at (0.75, -0.5) {};
		\node [style=none] (9) at (2.25, -0.5) {};
		\node [style=none] (10) at (1.25, 0.5) {};
		\node [style=none] (11) at (1.75, 0.45) {...};
		\node [style=none] (12) at (2.5, 0.75) {};
	\end{pgfonlayer}
	\begin{pgfonlayer}{edgelayer}
		\draw [style=braceedge] (6.center) to node[wire label, inner sep={5 pt}]{$m$} (2.center);
		\draw [style=none, bend left=15, looseness=1.00] (8.center) to (5);
		\draw [style=none, bend left=15, looseness=1.00] (0.center) to (5);
		\draw [style=braceedge] (4.center) to node[wire label, inner sep={5 pt}]{$n$} (12.center);
		\draw [style=none, bend left=15, looseness=1.00] (5) to (3.center);
		\draw [style=none, bend right=15, looseness=1.00] (9.center) to (5);
		\draw [style=none, bend right=15, looseness=1.00] (5) to (1.center);
		\draw [style=none, bend left=15, looseness=1.00] (5) to (10.center);
	\end{pgfonlayer}
\end{tikzpicture}}}\) & \(X^\alpha_{n,m}:n\to m\)\\[0.5em]
\hline
Hadamard & \(\vcenter{\hbox{\begin{tikzpicture}
	\begin{pgfonlayer}{nodelayer}
		\node [style={H box}] (0) at (0.5, 0) {};
		\node [style=none] (1) at (0.5, 0.3) {};
		\node [style=none] (2) at (0.5, -0.3) {};
				\node [style=none] (a) at (0.5, 0.5) {};
		\node [style=none] (b) at (0.5, -0.5) {};
	\end{pgfonlayer}
	\begin{pgfonlayer}{edgelayer}
		\draw (1.center) to (0);
		\draw (2.center) to (0);
	\end{pgfonlayer}
\end{tikzpicture}}}\) & \(H:1\to 1\)\\[0.5em]
\hline
Identity & \(\vcenter{\hbox{\begin{tikzpicture}
	\begin{pgfonlayer}{nodelayer}
		\node [style=none] (1) at (0.5, 0.3) {};
		\node [style=none] (2) at (0.5, -0.3) {};
	\end{pgfonlayer}
	\begin{pgfonlayer}{edgelayer}
		\draw (1.center) to (2.center);
	\end{pgfonlayer}
\end{tikzpicture}
}}\) & \(1_1:1\to 1\)\\[0.5em]
\hline
Braiding & \(\vcenter{\hbox{\begin{tikzpicture}
	\begin{pgfonlayer}{nodelayer}
		\node [style=none] (0) at (0, 0) {};
		\node [style=none] (1) at (0.3, 0.3) {};
		\node [style=none] (2) at (0.3, -0.3) {};
		\node [style=none] (a) at (-0.3, 0.3) {};
		\node [style=none] (b) at (-0.3, -0.3) {};
		\node [style=none] (3) at (0.3, -0.5) {};
		\node [style=none] (4) at (0.3, 0.5) {};
	\end{pgfonlayer}
	\begin{pgfonlayer}{edgelayer}
		\draw[bend left=35] (1.center) to (0.center);
		\draw[bend right=35] (2.center) to (0.center);
		\draw[bend right=35] (a.center) to (0.center);
		\draw[bend left=35] (b.center) to (0.center);
	\end{pgfonlayer}
\end{tikzpicture}}}\) & \(\sigma:2\to 2\)\\[0.5em]
\hline
Cup & \(\vcenter{\hbox{\begin{tikzpicture}
	\begin{pgfonlayer}{nodelayer}
		\node [style=none] (0) at (0, 0) {};
		\node [style=none] (1) at (0.3, 0.3) {};
		\node [style=none] (a) at (-0.3, 0.3) {};
	\end{pgfonlayer}
	\begin{pgfonlayer}{edgelayer}
		\draw[bend left=45] (1.center) to (0.center);
		\draw[bend right=45] (a.center) to (0.center);
	\end{pgfonlayer}
\end{tikzpicture}
}}\) & \(\eta:0\to 2\)\\[0.5em]
\hline
Cap & \(\vcenter{\hbox{\begin{tikzpicture}
	\begin{pgfonlayer}{nodelayer}
		\node [style=none] (0) at (0, 0) {};
		\node [style=none] (1) at (-0.3, -0.3) {};
		\node [style=none] (a) at (0.3, -0.3) {};
		\node [style=none] (3) at (0.2, -0.5) {};
		\node [style=none] (4) at (0.2, 0.5) {};
	\end{pgfonlayer}
	\begin{pgfonlayer}{edgelayer}
		\draw[bend left=45] (1.center) to (0.center);
		\draw[bend right=45] (a.center) to (0.center);
	\end{pgfonlayer}
\end{tikzpicture}}}\) & \(\epsilon:2\to 0\)\\[0.5em]
\hline
Unit (empty diagram) & \(\vcenter{\hbox{\begin{tikzpicture}
	\begin{pgfonlayer}{nodelayer}
		\node [style=none] (4) at (0.25, 0.25) {};
		\node [style=none] (5) at (-0.25, 0.25) {};
		\node [style=none] (6) at (0.25, -0.25) {};
		\node [style=none] (7) at (-0.25, -0.25) {};
	\end{pgfonlayer}
	\begin{pgfonlayer}{edgelayer}
		\draw [dashed, color=gray] (5.center) to (7.center);
		\draw [dashed, color=gray] (7.center) to (6.center);
		\draw [dashed, color=gray] (6.center) to (4.center);
		\draw [dashed, color=gray] (4.center) to (5.center);
	\end{pgfonlayer}
\end{tikzpicture}}}\) & \(e:0\to 0\)\\
\hline
\end{tabular}
\end{center}
\endgroup

We can compose diagrams in two ways:

\begin{itemize}
\item \textbf{Sequentially}: If \(D_1:a \to b\) and \(D_2:b \to c\) are two diagrams, we can compose \(D_2 \circ D_1 : a \to c\) by placing \(D_1\) above \(D_2\).  We can think of this new process as performing the process \(D_1\) first, and then applying the process \(D_2\) next.

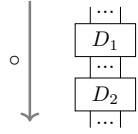
\widetikzitinput{"resources/Sequential.tikz"}

\item \textbf{Spatially}: If \(D_1:a\to b\) and \(D_2:c \to d\) are two diagrams, we can compose \(D_1 \ocross D_2: a+c \to b+d\) by placing \(D_1\) to the left of \(D_2\).  We can think of this new process as performing the processes \(D_1\) and \(D_2\) side-by-side.

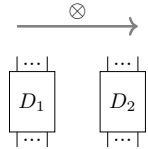
\widetikzitinput{"resources/Spacial.tikz"}
\end{itemize}

These operations satisfy the coherence conditions for a strict monoidal category.  See \cite{heunen2019categories} for details.

Importantly, we can define the \textbf{standard interpretation functor} \cite{coecke2011interacting}:

For any ZX diagram \(D:n\to m\), \(\intf{D}:(\mathbb C^2)^{\otimes n}\to(\mathbb C^2)^{\otimes m}\) is defined inductively in the following way:

\[
\intf{D_1\otimes D_2}
=\intf{D_1}\otimes\intf{D_2},
\quad
\intf{D_2\circ D_1}
=\intf{D_2}\circ\intf{D_1}.
\]

\[
\begin{array}{rcl@{\qquad\qquad}rcl}
\intf{\vcenter{\hbox{\begin{tikzpicture}
	\begin{pgfonlayer}{nodelayer}
		\node [style=none] (1) at (0.5, 0.3) {};
		\node [style=none] (2) at (0.5, -0.3) {};
	\end{pgfonlayer}
	\begin{pgfonlayer}{edgelayer}
		\draw (1.center) to (2.center);
	\end{pgfonlayer}
\end{tikzpicture}
}}}
&=&
\begin{pmatrix}
1 & 0\\
0 & 1
\end{pmatrix}
&
\intf{\vcenter{\hbox{\begin{tikzpicture}
	\begin{pgfonlayer}{nodelayer}
		\node [style=none] (0) at (0, 0) {};
		\node [style=none] (1) at (0.3, 0.3) {};
		\node [style=none] (2) at (0.3, -0.3) {};
		\node [style=none] (a) at (-0.3, 0.3) {};
		\node [style=none] (b) at (-0.3, -0.3) {};
		\node [style=none] (3) at (0.3, -0.5) {};
		\node [style=none] (4) at (0.3, 0.5) {};
	\end{pgfonlayer}
	\begin{pgfonlayer}{edgelayer}
		\draw[bend left=35] (1.center) to (0.center);
		\draw[bend right=35] (2.center) to (0.center);
		\draw[bend right=35] (a.center) to (0.center);
		\draw[bend left=35] (b.center) to (0.center);
	\end{pgfonlayer}
\end{tikzpicture}}}}
&=&
\begin{pmatrix}
1 & 0 & 0 & 0\\
0 & 0 & 1 & 0\\
0 & 1 & 0 & 0\\
0 & 0 & 0 & 1
\end{pmatrix}
\\[2em]
\intf{\vphantom{\vcenter{\hbox{\begin{tikzpicture}
	\begin{pgfonlayer}{nodelayer}
		\node [style=none] (0) at (0, 0) {};
		\node [style=none] (1) at (-0.3, -0.3) {};
		\node [style=none] (a) at (0.3, -0.3) {};
		\node [style=none] (3) at (0.2, -0.5) {};
		\node [style=none] (4) at (0.2, 0.5) {};
	\end{pgfonlayer}
	\begin{pgfonlayer}{edgelayer}
		\draw[bend left=45] (1.center) to (0.center);
		\draw[bend right=45] (a.center) to (0.center);
	\end{pgfonlayer}
\end{tikzpicture}}}}\vcenter{\hbox{\begin{tikzpicture}
	\begin{pgfonlayer}{nodelayer}
		\node [style=none] (0) at (0, 0) {};
		\node [style=none] (1) at (0.3, 0.3) {};
		\node [style=none] (a) at (-0.3, 0.3) {};
	\end{pgfonlayer}
	\begin{pgfonlayer}{edgelayer}
		\draw[bend left=45] (1.center) to (0.center);
		\draw[bend right=45] (a.center) to (0.center);
	\end{pgfonlayer}
\end{tikzpicture}
}}}
&=&
\ket{00}+\ket{11}
&
\intf{\vcenter{\hbox{\begin{tikzpicture}
	\begin{pgfonlayer}{nodelayer}
		\node [style=none] (0) at (0, 0) {};
		\node [style=none] (1) at (-0.3, -0.3) {};
		\node [style=none] (a) at (0.3, -0.3) {};
		\node [style=none] (3) at (0.2, -0.5) {};
		\node [style=none] (4) at (0.2, 0.5) {};
	\end{pgfonlayer}
	\begin{pgfonlayer}{edgelayer}
		\draw[bend left=45] (1.center) to (0.center);
		\draw[bend right=45] (a.center) to (0.center);
	\end{pgfonlayer}
\end{tikzpicture}}}}
&=&
\bra{00}+\bra{11}
\end{array}
\]
\[
\intf{\vcenter{\hbox{\begin{tikzpicture}\footnotesize
	\begin{pgfonlayer}{nodelayer}
		\node [style=none] (0) at (-0.5, 0.5) {};
		\node [style=none] (1) at (-1, 0.5) {};
		\node [style=none] (2) at (0, 0.45) {...};
		\node [style=none] (3) at (0, -0.45) {...};
		\node [style=none] (4) at (0.5, -0.5) {};
		\node [style=none] (5) at (0.75, 0.75) {};
		\node [style=none] (6) at (0.75, -0.75) {};
		\node [style=gn] (7) at (-0.25, -0) {$~\alpha~$};
		\node [style=none] (8) at (-1.25, 0.75) {};
		\node [style=none] (9) at (-1.25, -0.75) {};
		\node [style=none] (10) at (-1, -0.5) {};
		\node [style=none] (11) at (0.5, 0.5) {};
		\node [style=none] (12) at (-0.5, -0.5) {};
	\end{pgfonlayer}
	\begin{pgfonlayer}{edgelayer}
		\draw [style=braceedge] (6.center) to node[wire label, inner sep={5 pt}]{$m$} (9.center);
		\draw [style=none, bend left=15, looseness=1.00] (10.center) to (7);
		\draw [style=none, bend left=15, looseness=1.00] (12.center) to (7);
		\draw [style=braceedge] (8.center) to node[wire label, inner sep={5 pt}]{$n$} (5.center);
		\draw [style=none, bend left=15, looseness=1.00] (7) to (1.center);
		\draw [style=none, bend right=15, looseness=1.00] (4.center) to (7);
		\draw [style=none, bend right=15, looseness=1.00] (7) to (11.center);
		\draw [style=none, bend left=15, looseness=1.00] (7) to (0.center);
	\end{pgfonlayer}
\end{tikzpicture}}}}
=
\left\{
\begin{array}{rcl}
\ket{0\cdots 0} &\mapsto& \ket{0\cdots 0}\\
\ket{1\cdots 1} &\mapsto& e^{i\alpha}\ket{1\cdots 1}\\
\text{others} &\mapsto& 0
\end{array}
\right.
\]
\[
\intf{\vcenter{\hbox{\begin{tikzpicture}\footnotesize
	\begin{pgfonlayer}{nodelayer}
		\node [style=none] (0) at (1.25, -0.5) {};
		\node [style=none] (1) at (2.25, 0.5) {};
		\node [style=none] (2) at (0.5, -0.75) {};
		\node [style=none] (3) at (0.75, 0.5) {};
		\node [style=none] (4) at (0.5, 0.75) {};
		\node [style=rn] (5) at (1.5, 0) {$~\alpha~$};
		\node [style=none] (6) at (2.5, -0.75) {};
		\node [style=none] (7) at (1.75, -0.45) {...};
		\node [style=none] (8) at (0.75, -0.5) {};
		\node [style=none] (9) at (2.25, -0.5) {};
		\node [style=none] (10) at (1.25, 0.5) {};
		\node [style=none] (11) at (1.75, 0.45) {...};
		\node [style=none] (12) at (2.5, 0.75) {};
	\end{pgfonlayer}
	\begin{pgfonlayer}{edgelayer}
		\draw [style=braceedge] (6.center) to node[wire label, inner sep={5 pt}]{$m$} (2.center);
		\draw [style=none, bend left=15, looseness=1.00] (8.center) to (5);
		\draw [style=none, bend left=15, looseness=1.00] (0.center) to (5);
		\draw [style=braceedge] (4.center) to node[wire label, inner sep={5 pt}]{$n$} (12.center);
		\draw [style=none, bend left=15, looseness=1.00] (5) to (3.center);
		\draw [style=none, bend right=15, looseness=1.00] (9.center) to (5);
		\draw [style=none, bend right=15, looseness=1.00] (5) to (1.center);
		\draw [style=none, bend left=15, looseness=1.00] (5) to (10.center);
	\end{pgfonlayer}
\end{tikzpicture}}}}
=
\left\{
\begin{array}{rcl}
\ket{+\cdots +} &\mapsto& \ket{+\cdots +}\\
\ket{-\cdots -} &\mapsto& e^{i\alpha}\ket{-\cdots -}\\
\text{others} &\mapsto& 0
\end{array}
\right.
\]
\[
\intf{\vcenter{\hbox{\begin{tikzpicture}
	\begin{pgfonlayer}{nodelayer}
		\node [style={H box}] (0) at (0.5, 0) {};
		\node [style=none] (1) at (0.5, 0.3) {};
		\node [style=none] (2) at (0.5, -0.3) {};
				\node [style=none] (a) at (0.5, 0.5) {};
		\node [style=none] (b) at (0.5, -0.5) {};
	\end{pgfonlayer}
	\begin{pgfonlayer}{edgelayer}
		\draw (1.center) to (0);
		\draw (2.center) to (0);
	\end{pgfonlayer}
\end{tikzpicture}}}}
=\frac{1}{\sqrt{2}}
\begin{pmatrix}
1 & 1\\
1 & -1
\end{pmatrix}
\]

Now, we introduce the rule set \(\ZXsimp\), along with some facts previously established \cite{backens2020towards}.

\subsection{The Rule Set \texorpdfstring{\(\ZXsimp\)}{ZX\_simp}}
\label{sec:zx-simp}

We follow the phase convention of Backens--Perdrix--Wang throughout. Stabilizer phases are integer-labelled: the label \(k\in\mathbb Z\) denotes the angle \(k\pi/2\). Thus \(Z^k_{n,m}\) and \(X^k_{n,m}\) denote spiders with phase \(k\pi/2\). The addition in (S1) is ordinary integer addition, not addition modulo \(4\). In particular, \(k\) and \(k+4\) are not identified at the level of syntax. The \(2\pi\)-periodicity is a theorem of the full rule set of \cite{backens2020towards}, not a syntactic quotient that we impose here.

\begingroup
\footnotesize
\setlength{\tabcolsep}{3pt}
\renewcommand{\arraystretch}{1.35}
\begin{longtable}{@{}>{\centering\arraybackslash}m{0.08\textwidth}|>{\centering\arraybackslash}m{0.34\textwidth}|>{\centering\arraybackslash}m{0.52\textwidth}@{}}
\caption{\(\ZXsimp\)}
\label{tab:zx-simp}\\
\hline
Rule & Diagram & Typed version \\
\hline
\endfirsthead
\hline
Rule & Diagram & Typed version \\
\hline
\endhead
\hline
\multicolumn{3}{r@{}}{\emph{continued on next page}}\\
\endfoot
\hline
\endlastfoot
(S1) &
\begin{tikzpicture}[font={\footnotesize}]
	\begin{pgfonlayer}{nodelayer}
		\node [style=none] (1) at (0.25, -0) {$=$};
		\node [style=gn] (2) at (1.5, 0) { \footnotesize$\alpha{+}\beta$};
		\node [style=gn] (3) at (-0.75, -0.25) {\footnotesize$~\beta~$};
		\node [style=none] (4) at (-1.75, -0.5) {\raisebox{2mm}{...}};
		\node [style=none] (5) at (2, -0.75) {};
		\node [style=none] (6) at (-1, -0.75) {};
		\node [style=none] (7) at (1.5, -0.75) {\raisebox{2mm}{...}};
		\node [style=none] (8) at (-0.5, -0.75) {};
		\node [style=none] (9) at (1, -0.75) {};
		\node [style=none] (10) at (-2, -0.5) {};
		\node [style=none] (11) at (-1.5, -0.5) {};
		\node [style=none] (12) at (-0.75, -0.75) {\raisebox{2mm}{...}};
		\node [style=none] (13) at (1.5, 0.75) {\raisebox{-2mm}{...}};
		\node [style=none] (14) at (1, 0.75) {};
		\node [style=none] (15) at (-2, 0.75) {};
		\node [style=none] (16) at (-0.5, 0.5) {};
		\node [style=none] (17) at (-1.5, 0.75) {};
		\node [style=none] (18) at (2, 0.75) {};
		\node [style=gn] (19) at (-1.75, 0.25) {\footnotesize$~\alpha~$};
		\node [style=none] (20) at (-0.75, 0.5) {\raisebox{-2mm}{...}};
		\node [style=none] (21) at (-1, 0.5) {};
		\node [style=none] (22) at (-1.75, 0.75) {\raisebox{-2mm}{...}};
	\end{pgfonlayer}
	\begin{pgfonlayer}{edgelayer}
		\draw[bend right=23] (3) to (16.center);
		\draw[bend right=23] (3) to (6.center);
		\draw[bend left=23] (3) to (8.center);
		\draw[bend right=23] (19) to (10.center);
		\draw[bend left=23] (19) to (11.center);
		\draw%
		(19) to (3);
		\draw[bend right=23]  (14.center) to (2);
		\draw[bend right=23]  (2) to (9.center);
		\draw[bend right=23]  (5.center) to (2);
		\draw[bend right=23]  (2) to (18.center);
		\draw[bend left=23]  (19) to (15.center);
		\draw[bend right=23]  (19) to (17.center);
		\draw[bend left=23] (3) to (21.center);
	\end{pgfonlayer}
\end{tikzpicture} &
\(\begin{aligned}
&(1_q\otimes Z^\ell_{r+1,s})\circ (Z^k_{p,q+1}\otimes 1_r)
= Z^{k+\ell}_{p+r,q+s}: p+r\to q+s,\\
&p,q,r,s\in\mathbb N,\quad k,\ell\in\mathbb Z.
\end{aligned}\)
\\[0.5em]
(S3\(^{\prime}\)L) &
\(\cupDiagram=\greenCupDiagram\) &
\(\eta=Z^0_{0,2}:0\to 2.\)
\\[0.5em]
(S3\(^{\prime}\)R) &
\(\greenCupDiagram=\redCupDiagram\) &
\(Z^0_{0,2}=X^0_{0,2}:0\to 2.\)
\\[0.5em]
(B1) &
\begin{tikzpicture}
	\begin{pgfonlayer}{nodelayer}
		\node [style=gn] (0) at (0.75, 0) {};
		\node [style=none] (1) at (2.25, -0.25) {};
		\node [style=none] (2) at (0.5, -0.5) {};
		\node [style=rn] (3) at (2.25, 0.25) {};
		\node [style=none] (4) at (1, -0.5) {};
		\node [style=rn] (5) at (0.75, 0.5) {};
		\node [style=rn] (6) at (2.75, 0.25) {};
		\node [style=none] (7) at (2.75, -0.25) {};
		\node [style=none] (8) at (1.5, 0) {$=$};
		\node [style=rn] (9) at (0, 0.25) {};
		\node [style=gn] (10) at (0, -0.25) {};
	\end{pgfonlayer}
	\begin{pgfonlayer}{edgelayer}
		\draw [style=none] (5) to (0);
		\draw[bend right=23]  [style=none] (0) to (2.center);
		\draw[bend left=23]  [style=none] (0) to (4.center);
		\draw [style=none] (3) to (1.center);
		\draw [style=none] (6) to (7.center);
		\draw (9) to (10);
	\end{pgfonlayer}
\end{tikzpicture} &
\((Z^0_{1,0}\circ X^0_{0,1})\otimes
(Z^0_{1,2}\circ X^0_{0,1})=X^0_{0,1}\otimes X^0_{0,1}:0\to 2.\)
\\[0.5em]
(B2\(^{\prime}\)) &
\begin{tikzpicture}
	\begin{pgfonlayer}{nodelayer}
		\node [style=none] (0) at (3.75, 0.75) {};
		\node [style=rn] (1) at (0.5, -0.25) {};
		\node [style=none] (2) at (3.25, -0.75) {};
		\node [style=none] (3) at (1.25, 1) {};
		\node [style=none] (4) at (3.25, 0.75) {};
		\node [style=none] (5) at (0.5, -0.75) {};
		\node [style=none] (6) at (0.5, 1) {};
		\node [style=gn] (7) at (1.25, 0.5) {};
		\node [style=none] (8) at (2.25, 0) {$=$};
		\node [style=rn] (9) at (1.25, -0.25) {};
		\node [style=gn] (10) at (3.5, -0.25) {};
		\node [style=gn] (11) at (0.5, 0.5) {};
		\node [style=none] (12) at (1.25, -0.75) {};
		\node [style=none] (13) at (3.75, -0.75) {};
		\node [style=rn] (14) at (3.5, 0.25) {};
		\node [style=rn] (15) at (3, 0.25) {};
		\node [style=gn] (16) at (3, -0.25) {};
		\node [style=gn] (17) at (0, -0) {};
	\end{pgfonlayer}
	\begin{pgfonlayer}{edgelayer}
		\draw [style=none] (12.center) to (9);
		\draw [style=none] (5.center) to (1);
		\draw [style=none] (7) to (3.center);
		\draw [style=none, bend right=23, looseness=1.00] (9) to (7);
		\draw [style=none] (11) to (6.center);
		\draw [style=none, bend left=23, looseness=1.00] (1) to (11);
		\draw [style=none, bend right=23, looseness=1.00] (13.center) to (10);
		\draw [style=none] (10) to (14);
		\draw [style=none, bend left=23, looseness=1.00] (14) to (4.center);
		\draw [style=none, bend right=23, looseness=1.00] (14) to (0.center);
		\draw [bend right=23, looseness=1.00] (10) to (2.center);
		\draw (11) to (9);
		\draw (7) to (1);
		\draw (15) to (16);
	\end{pgfonlayer}
\end{tikzpicture} &
\(\begin{aligned}
&Z^0_{0,0}\otimes\big((X^0_{2,1}\otimes X^0_{2,1})
\circ(1_1\otimes\sigma\otimes 1_1)
\circ(Z^0_{1,2}\otimes Z^0_{1,2})\big)\\
&\qquad =(Z^0_{1,0}\circ X^0_{0,1})\otimes
(Z^0_{1,2}\circ X^0_{2,1}):2\to 2.
\end{aligned}\)
\\[0.5em]
(EU\(^{\prime}\)) &
\begin{tikzpicture}
	\begin{pgfonlayer}{nodelayer}
		\node [style=gn] (0) at (2.25, 0.5) {$\nicefrac{\pi}{2}$};
		\node [style=gn] (9) at (3, 0.25) {$\nicefrac{\textnormal{-}\pi}{2}$};
		\node [style=rn] (1) at (2.25, -0) {};
		\node [style=none] (2) at (0.5, 0.5) {};
		\node [style=none] (3) at (0.5, -0.5) {};
		\node [style=none] (4) at (2.25, -1) {};
		\node [style=none] (5) at (1.25, 0) {$=$};
		\node [style={H box}] (6) at (0.5, -0) {};
		\node [style=gn] (7) at (2.25, -0.5) {$\nicefrac{\pi}{2}$};
		\node [style=none] (8) at (2.25, 1) {};
	\end{pgfonlayer}
	\begin{pgfonlayer}{edgelayer}
		\draw (2.center) to (6);
		\draw (3.center) to (6);
		\draw (8.center) to (0);
		\draw (0) to (1);
		\draw (1) to (7);
		\draw (7) to (4.center);
		\draw (9) to (1);
	\end{pgfonlayer}
\end{tikzpicture} &
\(H=Z^1_{1,1}\circ X^0_{2,1}\circ
(Z^1_{1,1}\otimes Z^{-1}_{0,1}):1\to 1.\)
\\[0.5em]
(H) &
\begin{tikzpicture}
	\begin{pgfonlayer}{nodelayer}
		\node [style=none] (0) at (1, 0.75) {};
		\node [style=none] (1) at (-1.25, -1) {};
		\node [style={H box}] (2) at (-1.25, 0.5)  {};
		\node [style=none] (3) at (-1.75, -0.75) {\raisebox{2mm}{...}};
		\node [style=none] (4) at (1.5, -0.75) {\raisebox{2mm}{...}};
		\node [style=none] (5) at (1.5, 0.75) {\raisebox{-2mm}{...}};
		\node [style=rn] (6) at (-1.75, -0) {\footnotesize$~\alpha~$};
		\node [style={H box}] (7) at (-2.25, -0.5)  {};
		\node [style=none] (8) at (-1.25, 1) {};
		\node [style={H box}] (9) at (-1.25, -0.5)  {};
		\node [style=none] (10) at (-1.75, 0.75) {\raisebox{-2mm}{...}};
		\node [style=none] (11) at (0, -0) {$=$};
		\node [style=gn] (12) at (1.5, -0) {\footnotesize$~\alpha~$};
		\node [style=none] (13) at (2, -0.75) {};
		\node [style=none] (14) at (2, 0.75) {};
		\node [style=none] (15) at (-2.25, 1) {};
		\node [style={H box}] (16) at (-2.25, 0.5) {};
		\node [style=none] (17) at (-2.25, -1) {};
		\node [style=none] (18) at (1, -0.75) {};
	\end{pgfonlayer}
	\begin{pgfonlayer}{edgelayer}
		\draw[bend right] (6) to (7);
		\draw[bend left] (6) to (9);
		\draw (9) to (1.center);
		\draw (7) to (17.center);
		\draw[bend left] (6) to (16);
		\draw[bend right] (6) to (2);
		\draw (2) to (8.center);
		\draw (16) to (15.center);
		\draw[bend left=23]  (12) to (0.center);
		\draw[bend right=23]  (12) to (14.center);
		\draw[bend right=23]  (12) to (18.center);
		\draw[bend left=23]  (12) to (13.center);
	\end{pgfonlayer}
\end{tikzpicture} &
\(\begin{aligned}
H^{\otimes m}\circ X^k_{n,m}\circ H^{\otimes n}
&=Z^k_{n,m}:n\to m,\quad n,m\in\mathbb N,\quad k\in\mathbb Z.
\end{aligned}\)
\\[0.5em]
(IV\(^{\prime}\)) &
\begin{tikzpicture}
	\begin{pgfonlayer}{nodelayer}
		\node [style=gn] (0) at (-0.25, -0.25) {};
		\node [style=rn] (1) at (-0.25, 0.5) {};
		\node [style=gn] (2) at (0.25, -0.25) {};
		\node [style=rn] (3) at (0.25, 0.5) {};
		\node [style=none] (4) at (0.75, 0) {$=$};
		\node [style=none] (5) at (1.75, 0.25) {};
		\node [style=none] (6) at (1.75, -0.25) {};
		\node [style=none] (7) at (1.25, 0.25) {};
		\node [style=none] (8) at (1.25, -0.25) {};
		\node [style=gn] (9) at (-0.75, -0) {};
	\end{pgfonlayer}
	\begin{pgfonlayer}{edgelayer}
		\draw [bend left=45, looseness=1.00] (1) to (0);
		\draw [bend right=45, looseness=1.00] (1) to (0);
		\draw (1) to (0);
		\draw [bend left=45, looseness=1.00] (3) to (2);
		\draw [bend right=45, looseness=1.00] (3) to (2);
		\draw (3) to (2);
		\draw [color=gray, dashed] (7.center) to (8.center);
		\draw [color=gray, dashed] (8.center) to (6.center);
		\draw [color=gray, dashed] (6.center) to (5.center);
		\draw [color=gray, dashed] (5.center) to (7.center);
	\end{pgfonlayer}
\end{tikzpicture} &
\(Z^0_{0,0}\otimes
(Z^0_{3,0}\circ X^0_{0,3})\otimes
(Z^0_{3,0}\circ X^0_{0,3})=e:0\to 0.\)
\\[0.5em]
(ZO\(^{\prime}\)) &
\begin{tikzpicture}
	\begin{pgfonlayer}{nodelayer}
		\node [style=gn] (0) at (-1, 0) {$~\pi~$};
		\node [style=none] (2) at (-0.5, -0.25) {};
		\node [style=none] (3) at (0, 0) {$=$};
		\node [style=gn] (4) at (-0.5, 0.25) {};
		\node [style=rn] (6) at (1, 0.25) {};
		\node [style=none] (7) at (1, -0.25) {};
		\node [style=gn] (8) at (0.5, 0) {$~\pi~$};
	\end{pgfonlayer}
	\begin{pgfonlayer}{edgelayer}
		\draw (2) to (4);
		\draw (6) to (7.center);
	\end{pgfonlayer}
\end{tikzpicture} &
\(\begin{aligned}
Z^2_{0,0}\otimes Z^0_{0,1}
=Z^2_{0,0}\otimes X^0_{0,1}:0\to 1.
\end{aligned}\)
\end{longtable}
\endgroup

\begin{definition}[Meta-rule: only connectivity matters]
Two diagrams represent the same matrix whenever one can be transformed into the other by moving components around without changing their connections.
\end{definition}

\begin{theorem}[\cite{backens2020towards}, Theorem~2.10]
\(\ZXsimp\) is sound and complete, i.e. for any two stabilizer ZX-calculus diagrams \(D_1\) and \(D_2\), we have:
\[
\ZXsimp\vdash D_1=D_2\quad\Longleftrightarrow\quad \intf{D_1}=\intf{D_2}.
\]
\end{theorem}

\begin{theorem}[\cite{backens2020towards}, Lemmas~3.1--3.8]
Every rule in \(\ZXsimp\setminus\{(\mathrm{S3}'\mathrm{R}),(\mathrm{B2}')\}\) is necessary, and at least one of (S3'R), (B2') is necessary.
\end{theorem}

As we will see, it turns out that both of these rules are individually necessary, meaning every rule in \(\ZXsimp\) is necessary relative to the connectivity meta-rule of Backens--Perdrix--Wang \cite{backens2020towards}.

\section{On the Necessity of (S3'R)}
\label{sec:s3r-necessity}

S3'R is necessary, as we argue below.

\begin{theorem}
\label{thm:s3r-necessary}
The (S3'R) rule is necessary: \(\ZXsimp\setminus\{(\mathrm{S3}'\mathrm{R})\}\nvdash(\mathrm{S3}'\mathrm{R})\).
\end{theorem}

A countermodel here means an interpretation satisfying every non-target rule while falsifying the target equation.

\begin{proof}
Work in ordinary complex matrices on qubits, with object \(1\) interpreted as \(\mathbb C^2\).

Let \(\intf{-}\) be the standard interpretation functor.

For any diagram \(D:n\to m\), define a new interpretation \(\intf{D}^{(\mathrm{S3}'\mathrm{R})}\in\mathbb C^{2^m\times 2^n}\) inductively:
\[
\intf{D_1\otimes D_2}^{(\mathrm{S3}'\mathrm{R})}
=\intf{D_1}^{(\mathrm{S3}'\mathrm{R})}\otimes\intf{D_2}^{(\mathrm{S3}'\mathrm{R})},
\quad
\intf{D_2\circ D_1}^{(\mathrm{S3}'\mathrm{R})}
=\intf{D_2}^{(\mathrm{S3}'\mathrm{R})}\circ\intf{D_1}^{(\mathrm{S3}'\mathrm{R})},
\]
\[
\intf{\vcenter{\hbox{\begin{tikzpicture}\footnotesize
	\begin{pgfonlayer}{nodelayer}
		\node [style=none] (0) at (-0.5, 0.5) {};
		\node [style=none] (1) at (-1, 0.5) {};
		\node [style=none] (2) at (0, 0.45) {...};
		\node [style=none] (3) at (0, -0.45) {...};
		\node [style=none] (4) at (0.5, -0.5) {};
		\node [style=none] (5) at (0.75, 0.75) {};
		\node [style=none] (6) at (0.75, -0.75) {};
		\node [style=gn] (7) at (-0.25, -0) {$~\alpha~$};
		\node [style=none] (8) at (-1.25, 0.75) {};
		\node [style=none] (9) at (-1.25, -0.75) {};
		\node [style=none] (10) at (-1, -0.5) {};
		\node [style=none] (11) at (0.5, 0.5) {};
		\node [style=none] (12) at (-0.5, -0.5) {};
	\end{pgfonlayer}
	\begin{pgfonlayer}{edgelayer}
		\draw [style=braceedge] (6.center) to node[wire label, inner sep={5 pt}]{$m$} (9.center);
		\draw [style=none, bend left=15, looseness=1.00] (10.center) to (7);
		\draw [style=none, bend left=15, looseness=1.00] (12.center) to (7);
		\draw [style=braceedge] (8.center) to node[wire label, inner sep={5 pt}]{$n$} (5.center);
		\draw [style=none, bend left=15, looseness=1.00] (7) to (1.center);
		\draw [style=none, bend right=15, looseness=1.00] (4.center) to (7);
		\draw [style=none, bend right=15, looseness=1.00] (7) to (11.center);
		\draw [style=none, bend left=15, looseness=1.00] (7) to (0.center);
	\end{pgfonlayer}
\end{tikzpicture}}}}^{(\mathrm{S3}'\mathrm{R})}
:= i^{n+m-2}\cdot \intf{\vcenter{\hbox{\begin{tikzpicture}\footnotesize
	\begin{pgfonlayer}{nodelayer}
		\node [style=none] (0) at (-0.5, 0.5) {};
		\node [style=none] (1) at (-1, 0.5) {};
		\node [style=none] (2) at (0, 0.45) {...};
		\node [style=none] (3) at (0, -0.45) {...};
		\node [style=none] (4) at (0.5, -0.5) {};
		\node [style=none] (5) at (0.75, 0.75) {};
		\node [style=none] (6) at (0.75, -0.75) {};
		\node [style=gn] (7) at (-0.25, -0) {$~\alpha~$};
		\node [style=none] (8) at (-1.25, 0.75) {};
		\node [style=none] (9) at (-1.25, -0.75) {};
		\node [style=none] (10) at (-1, -0.5) {};
		\node [style=none] (11) at (0.5, 0.5) {};
		\node [style=none] (12) at (-0.5, -0.5) {};
	\end{pgfonlayer}
	\begin{pgfonlayer}{edgelayer}
		\draw [style=braceedge] (6.center) to node[wire label, inner sep={5 pt}]{$m$} (9.center);
		\draw [style=none, bend left=15, looseness=1.00] (10.center) to (7);
		\draw [style=none, bend left=15, looseness=1.00] (12.center) to (7);
		\draw [style=braceedge] (8.center) to node[wire label, inner sep={5 pt}]{$n$} (5.center);
		\draw [style=none, bend left=15, looseness=1.00] (7) to (1.center);
		\draw [style=none, bend right=15, looseness=1.00] (4.center) to (7);
		\draw [style=none, bend right=15, looseness=1.00] (7) to (11.center);
		\draw [style=none, bend left=15, looseness=1.00] (7) to (0.center);
	\end{pgfonlayer}
\end{tikzpicture}}}},
\]
\[
\intf{{\let\alpha\beta\vcenter{\hbox{\begin{tikzpicture}\footnotesize
	\begin{pgfonlayer}{nodelayer}
		\node [style=none] (0) at (1.25, -0.5) {};
		\node [style=none] (1) at (2.25, 0.5) {};
		\node [style=none] (2) at (0.5, -0.75) {};
		\node [style=none] (3) at (0.75, 0.5) {};
		\node [style=none] (4) at (0.5, 0.75) {};
		\node [style=rn] (5) at (1.5, 0) {$~\alpha~$};
		\node [style=none] (6) at (2.5, -0.75) {};
		\node [style=none] (7) at (1.75, -0.45) {...};
		\node [style=none] (8) at (0.75, -0.5) {};
		\node [style=none] (9) at (2.25, -0.5) {};
		\node [style=none] (10) at (1.25, 0.5) {};
		\node [style=none] (11) at (1.75, 0.45) {...};
		\node [style=none] (12) at (2.5, 0.75) {};
	\end{pgfonlayer}
	\begin{pgfonlayer}{edgelayer}
		\draw [style=braceedge] (6.center) to node[wire label, inner sep={5 pt}]{$m$} (2.center);
		\draw [style=none, bend left=15, looseness=1.00] (8.center) to (5);
		\draw [style=none, bend left=15, looseness=1.00] (0.center) to (5);
		\draw [style=braceedge] (4.center) to node[wire label, inner sep={5 pt}]{$n$} (12.center);
		\draw [style=none, bend left=15, looseness=1.00] (5) to (3.center);
		\draw [style=none, bend right=15, looseness=1.00] (9.center) to (5);
		\draw [style=none, bend right=15, looseness=1.00] (5) to (1.center);
		\draw [style=none, bend left=15, looseness=1.00] (5) to (10.center);
	\end{pgfonlayer}
\end{tikzpicture}}}}}^{(\mathrm{S3}'\mathrm{R})}
:= -1\cdot \intf{{\let\alpha\beta\vcenter{\hbox{\begin{tikzpicture}\footnotesize
	\begin{pgfonlayer}{nodelayer}
		\node [style=none] (0) at (1.25, -0.5) {};
		\node [style=none] (1) at (2.25, 0.5) {};
		\node [style=none] (2) at (0.5, -0.75) {};
		\node [style=none] (3) at (0.75, 0.5) {};
		\node [style=none] (4) at (0.5, 0.75) {};
		\node [style=rn] (5) at (1.5, 0) {$~\alpha~$};
		\node [style=none] (6) at (2.5, -0.75) {};
		\node [style=none] (7) at (1.75, -0.45) {...};
		\node [style=none] (8) at (0.75, -0.5) {};
		\node [style=none] (9) at (2.25, -0.5) {};
		\node [style=none] (10) at (1.25, 0.5) {};
		\node [style=none] (11) at (1.75, 0.45) {...};
		\node [style=none] (12) at (2.5, 0.75) {};
	\end{pgfonlayer}
	\begin{pgfonlayer}{edgelayer}
		\draw [style=braceedge] (6.center) to node[wire label, inner sep={5 pt}]{$m$} (2.center);
		\draw [style=none, bend left=15, looseness=1.00] (8.center) to (5);
		\draw [style=none, bend left=15, looseness=1.00] (0.center) to (5);
		\draw [style=braceedge] (4.center) to node[wire label, inner sep={5 pt}]{$n$} (12.center);
		\draw [style=none, bend left=15, looseness=1.00] (5) to (3.center);
		\draw [style=none, bend right=15, looseness=1.00] (9.center) to (5);
		\draw [style=none, bend right=15, looseness=1.00] (5) to (1.center);
		\draw [style=none, bend left=15, looseness=1.00] (5) to (10.center);
	\end{pgfonlayer}
\end{tikzpicture}}}}},
\]
\[
\intf{\vcenter{\hbox{\begin{tikzpicture}
	\begin{pgfonlayer}{nodelayer}
		\node [style={H box}] (0) at (0.5, 0) {};
		\node [style=none] (1) at (0.5, 0.3) {};
		\node [style=none] (2) at (0.5, -0.3) {};
				\node [style=none] (a) at (0.5, 0.5) {};
		\node [style=none] (b) at (0.5, -0.5) {};
	\end{pgfonlayer}
	\begin{pgfonlayer}{edgelayer}
		\draw (1.center) to (0);
		\draw (2.center) to (0);
	\end{pgfonlayer}
\end{tikzpicture}}}}^{(\mathrm{S3}'\mathrm{R})}
:= i\cdot \intf{\vcenter{\hbox{\begin{tikzpicture}
	\begin{pgfonlayer}{nodelayer}
		\node [style={H box}] (0) at (0.5, 0) {};
		\node [style=none] (1) at (0.5, 0.3) {};
		\node [style=none] (2) at (0.5, -0.3) {};
				\node [style=none] (a) at (0.5, 0.5) {};
		\node [style=none] (b) at (0.5, -0.5) {};
	\end{pgfonlayer}
	\begin{pgfonlayer}{edgelayer}
		\draw (1.center) to (0);
		\draw (2.center) to (0);
	\end{pgfonlayer}
\end{tikzpicture}}}}.
\]
We leave \(I,\sigma,\eta,\epsilon,e\) unchanged.

\begin{lemma}
For any ZX diagram \(D\),
\[
\intf{D}^{(\mathrm{S3}'\mathrm{R})}=c(D)\cdot\intf{D},
\]
where
\begin{equation}
\label{eq:scalar-cD}
c(D)=(i)^{\sum_{v\in Z(D)}(\deg(v)-2)}\cdot(-1)^{|X(D)|}\cdot(i)^{|H(D)|}.
\end{equation}
Here \(\deg(v)\) is the number of incident half-edges, with self-loops counted twice, and \(Z(D)\), \(X(D)\), and \(H(D)\) are the collections of \(Z\)-spider, \(X\)-spider, and Hadamard vertices in the diagram.
\end{lemma}

\begin{proof}
See \Cref{app:proof-lemma-3-2}.
\end{proof}

The scalar \(c(D)\) is invariant under graph isomorphism, so the connectivity meta-rule is sound.

Now, we can just compare the scalar \(c(D)\) on both sides of each rewrite, since the standard interpretation agrees with each of them.

\[
\begin{array}{@{}c|c@{}}
\toprule
\text{Rule} & \text{Scalar (verified on both sides)}\\
\midrule
(\mathrm{S1}) & i^{p+q+r+s-2}\\
(\mathrm{S3}'\mathrm{L}) & 1\\
(\mathrm{B1}) & 1\\
(\mathrm{B2}') & 1\\
(\mathrm{EU}') & i\\
(\mathrm{H}) & -1\cdot i^{n+m}\\
(\mathrm{IV}') & 1\\
(\mathrm{ZO}') & \text{both sides zero}\\
\bottomrule
\end{array}
\]

However, \(\intf{\cupDiagram}^{(\mathrm{S3}'\mathrm{R})}=\eta\), and \(\intf{\redCupDiagram}^{(\mathrm{S3}'\mathrm{R})}=-\eta\). Hence, (S3'R) is necessary.
\end{proof}

Physically, one can interpret this result as saying that we cannot assume the compact structure on the X-spider holds a priori.

\section{The Road to Bialgebra}\label{sec:road}

In this section, we derive the interpretation used for necessity of the Bialgebra law. For notational clarity, we use $H$ to mean the interpretation of the Hadamard, and $H_0$ to be the usual unnormalized $4\times4$ variant. Likewise, $Z^k_{n,m}$ and $X^k_{n,m}$ represent the $Z$ and $X$ interpretations throughout this section.

In this section, we will outline the process for deriving the bialgebra-separating interpretation presented in \Cref{sec:b2-necessity}, which is intended to be self-contained.

\subsection{The Target Category}

First, we will introduce our target category, and motivate these choices.

Consider the ring of dual numbers over $\F_5$,
\[
R=\F_5[\delta]/(\delta^2).
\]
Here, $\delta\neq0$ and $\delta^2=0$. One should think of $\delta$ as an infinitesimal element which we have added to $\F_5$. Note that $R$ is commutative.

Now, consider the free $R$-module $V=R^4$ with canonical basis $e_0,e_1,e_2,e_3$. Finally, define our target category $\mathcal C$ as follows. It is the strict symmetric monoidal category whose objects are natural numbers, with
\[
\mathcal C(n,m)=\Hom_R\!\left(V^{\otimes n},V^{\otimes m}\right),
\qquad V^{\otimes0}=R.
\]
Composition is ordinary composition of $R$-linear maps, and parallel composition is the usual Kronecker tensor product. We use the standard ordered-basis identification
\[
V^{\otimes n}\otimes V^{\otimes m}=V^{\otimes(n+m)}.
\]

Our structural generators are defined as
\[
\begin{array}{rcl}
\text{identity:}&I&=\id_V,\\
\text{braiding:}&\sigma(v\otimes w)&=w\otimes v,\\
\text{cup:}&\eta(1)&=\displaystyle\sum_{r=0}^{3}e_r\otimes e_r,\\
\text{cap:}&\epsilon(e_r\otimes e_s)&=\delta_{rs},\\
\text{unit:}&e&=1_R.
\end{array}
\]
These maps give the usual compact structure: for each basis vector $e_a$,
\[
(\epsilon\otimes I)(I\otimes\eta)(e_a)=e_a,
\qquad
(I\otimes\epsilon)(\eta\otimes I)(e_a)=e_a.
\]
Because \(V\) is finite free, the usual basis cup and cap give the expected compact structure.

We want to think of $\mathcal C$ as a way to generalize a qubit to four dimensions, rather than two, and to work over a finite field rather than the complex numbers, along with a new infinitesimal element.

Then, the motivation to work over both a finite field and to introduce the infinitesimal is similar. Both features introduce simplifications to the algebraic constraints imposed by the rules, particularly $(EU')$. The element $\delta$ gives the most impactful simplification, since it turns many complicated nonlinear equations into linear ones, as we will soon see, since all terms higher than $\delta^2$ reduce to $0$.

\subsection{Green spiders}

Now, we will introduce our interpretations of the Hadamard, Green, and Red spiders in our category. These are all defined analogously to their standard interpretation, augmented to dimension $4$.

For each $k\in\Z$, define
\[
Z^k_{n,m}:V^{\otimes n}\longrightarrow V^{\otimes m}
\]
by
\[
Z^k_{n,m}=\sum_{r=0}^{3}t_r^k
\ket{e_r^{\otimes m}}\bra{e_r^{\otimes n}},
\qquad e_r^{\otimes0}=1.
\]
We leave the $t_r$ unspecified, but we see that they must be units of $R$ to satisfy $(S1)$.

We write
\[
t=a+\delta p,
\qquad a,p\in\F_5^4.
\]
In the standard interpretation,
\[
P(\pi/2)=\diag(1,\iota),
\qquad \iota^2=-1.
\]
So, in our doubled-dimensional interpretation over $\F_5$,
\[
P(\pi/2)\otimes P(\pi/2)=\diag(1,\iota,\iota,-1).
\]
Since $3^2=-1$ in $\F_5$, we choose $\iota=3$. Thus, a good choice is
\[
a=(1,3,3,4),
\]
and it remains only to solve $p$, so
\[
t=
\begin{pmatrix}
1+\delta p_0\\
3+\delta p_1\\
3+\delta p_2\\
4+\delta p_3
\end{pmatrix}.
\]
Furthermore, we can calculate the inverse tuple:
\[
t^{-1}=
\begin{pmatrix}
1+4\delta p_0\\
2+\delta p_1\\
2+\delta p_2\\
4+4\delta p_3
\end{pmatrix}.
\]

\subsection{The Hadamard}

Write
\[
H=\frac12\left(H_0+\delta H_1\right),
\qquad
H_0=
\begin{pmatrix}
1&1\\
1&4
\end{pmatrix}^{\!\otimes2},
\qquad
H_1\in M_4(\F_5).
\]
Incidentally, this is why we chose to work in the field \(\mathbb F_5\) rather than \(\mathbb F_2\) or \(\mathbb F_3\): it is the first finite field where \(2\) is invertible and \(-1\) has a square root.

Here \(4=-1\) in \(\F_5\). Thus \(H_0\) is the usual
\(4\times4\) Sylvester, or Walsh--Hadamard, matrix
\cite[Section~2.1]{cameron_hadamard}.
Equivalently, after identifying \(4\) with \(-1\), this is the character
table of the elementary abelian group \((\mathbb Z/2)^2\).

The term \(\delta H_1\) should be viewed as a first-order deformation
of this ordinary Walsh--Hadamard matrix.

Now, we impose the following constraints on the Hadamard and phase tuple, in order to satisfy connectivity and the structural rules. We require that the Hadamard is symmetric and involutive:
\[
H^T=H,
\qquad
H^2=I.
\]
We already know that $H_0^T=H_0$, so
\[
H_1^T=H_1.
\]
Also, since $H_0^2=4I$ and $\delta^2=0$,
\[
H^2
=\frac14\left(H_0^2+\delta(H_0H_1+H_1H_0)\right)
=I+\frac14\delta(H_0H_1+H_1H_0).
\]
And, since $\delta\neq0$,
\[
H_0H_1+H_1H_0=0.
\]
Intuitively, this construction is such that the matrix \(H_0+\delta H_1\)
produces an involutive \(H=\frac12(H_0+\delta H_1)\). The later failure of \((B2')\)
is therefore detected only in the first-order \(\delta\)-part of the
model.

\subsection{Red spiders}

We simply give the Red spiders by Hadamard conjugation, as usual:
\[
X^k_{n,m}=H^{\otimes m}\circ Z^k_{n,m}\circ H^{\otimes n}.
\]

\subsection{The structural rules}

Assuming such $p$ and $H_1$ exist, we have the following rules satisfied:
\[
\text{connectivity},
\qquad
(S1),
\qquad
(\mathrm{S3}'\mathrm{L}),
\qquad
(\mathrm{S3}'\mathrm{R}),
\qquad
(H).
\]
These all hold by construction. See \Cref{sec:b2-necessity} and the \hyperref[app:proof-theorem-5-6]{appendix} for detailed computations.

Now it remains to extract constraints from the remaining rules, and crucially $(B2')$, to specify the Hadamard and phases.

\subsection{The copy rule $(B1)$}

The copy rule will determine the Hadamard up to a scalar factor.

First, we can translate the copy $(B1)$ rule into an algebraic equation. Writing
\[
u=e_0+e_1+e_2+e_3,
\qquad
\Delta(e_r)=e_r\otimes e_r,
\qquad
\rho=Hu=X^0_{0,1},
\]
we have
\[
(u^T\rho)\,\Delta(\rho)=\rho\otimes\rho.
\]
It is easy to check that this is satisfied if and only if
\[
H_1u=0.
\]
Equivalently, $(B1)$ is satisfied if and only if every row and column sum of $H_1$ vanishes.

This is enough to determine $H_1$ up to a scalar. Write
\[
H_1=
\begin{pmatrix}
a&b&c&d\\
b&e&f&g\\
c&f&h&i\\
d&g&i&j
\end{pmatrix}.
\]
Since $H_0H_1+H_1H_0=0$ and $H_1u=0$, we get
\[
a=b=c=d=f=j=0,
\qquad
 i=e,
\qquad
 g=h=-e=: \lambda\in\F_5.
\]
Then,
\[
H_1=\lambda
\begin{pmatrix}
0&0&0&0\\
0&4&0&1\\
0&0&1&4\\
0&1&4&0
\end{pmatrix}.
\]
Thus, we have $H$ determined up to some $\lambda\in\F_5$.

\subsection{The Euler rule $(EU')$}

Now, we can perform the same process for $(EU')$, which similarly specifies our phases up to scalar. We write
\[
D=\diag(t_0,t_1,t_2,t_3)
\]
and note that the green multiplication acts by coordinatewise multiplication:
\[
Z^0_{2,1}(v\otimes w)=\sum_{r=0}^{3}v_rw_r e_r.
\]
Thus, the matrix equation for $(EU')$ is
\[
D H\,\diag(Ht^{-1})\,H D=H.
\]
Substituting $D$, $H$, and $t$, and comparing both sides, we extract the following equations from the diagonal:
\[
\begin{cases}
3p_0=0,\\
3p_0+3p_1+4\lambda=0,\\
3p_0+3p_2+\lambda=0,\\
2p_0+4p_3=0.
\end{cases}
\]
Therefore, we can specify our phases up to the Hadamard scalar $\lambda$:
\[
p_0=0,
\qquad
p_1=2\lambda,
\qquad
p_2=3\lambda,
\qquad
p_3=0.
\]
Therefore, we obtain the phase and its inverse:
\[
t=
\begin{pmatrix}
1\\
3+2\lambda\delta\\
3+3\lambda\delta\\
4
\end{pmatrix},
\qquad
t^{-1}=
\begin{pmatrix}
1\\
2+2\lambda\delta\\
2+3\lambda\delta\\
4
\end{pmatrix},
\qquad \lambda\in\F_5.
\]
One can check that both $(ZO')$ and $(IV')$ hold for any $\lambda\in\F_5$.

\subsection{The bialgebra rule $(B2')$}
\label{sec:bialgebra-rule}

Finally, it remains to choose such a $\lambda$ which breaks the bialgebra law. We write
\[
F:=X^0_{2,1}=H\circ Z^0_{2,1}\circ(H\otimes H)
\]
for the Red multiplication, and $\Delta=Z^0_{1,2}$. Then, our matrix equation for the bialgebra is
\[
4(F\otimes F)(I\otimes\sigma\otimes I)(\Delta\otimes\Delta)
=2\Delta F.
\]
Now, evaluating both sides on $e_1\otimes e_1$ and comparing the $e_0\otimes e_2$ coefficients\daggerfootnote{See the \hyperref[app:bialgebra-matrices]{appendix} for the full \(16\times16\) matrices. The \((e_1\otimes e_1,e_0\otimes e_2)\) coordinate is just one convenient coordinate for comparison.}, we find that
\[
\lambda\delta=0.
\]
Now, $\delta\neq0$, so for bialgebra to hold, $\lambda=0$. Therefore, we can choose any $\lambda\neq0$, and we take $\lambda=1$ for simplicity.

So, we find a Hadamard and phase tuple which suffice:
\[
H=\frac12(H_0+\delta H_1),
\qquad
H_1=
\begin{pmatrix}
0&0&0&0\\
0&4&0&1\\
0&0&1&4\\
0&1&4&0
\end{pmatrix},
\]
\[
t=(1,3+2\delta,3+3\delta,4),
\qquad
p=(0,2,3,0).
\]

In the next section, we introduce the formally defined interpretation and compute each axiom.

\section{On the Necessity of (B2')}\label{sec:b2-necessity}

We work in the strict symmetric monoidal category $\mathcal C$ outlined in \Cref{sec:road}. Explicitly, let
\[
R:=\F_5[\delta]/(\delta^2)
\]
with $\delta\neq0$ and $\delta^2=0$, and let $V:=R^4$ with canonical basis $e_0,e_1,e_2,e_3$. The objects of $\mathcal C$ are natural numbers, with
\[
\mathcal C(n,m)=\Hom_R\!\left(V^{\otimes n},V^{\otimes m}\right),
\qquad V^{\otimes0}=R.
\]
Composition is ordinary composition of $R$-linear maps, and parallel composition is the usual Kronecker tensor product. We use the standard ordered-basis identification
\[
V^{\otimes n}\otimes V^{\otimes m}=V^{\otimes(n+m)}.
\]
The structural maps are
\[
\begin{aligned}
I&=\id_V,
&\sigma(v\otimes w)&=w\otimes v,\\
\eta(1)&=\sum_{r=0}^{3}e_r\otimes e_r,
&\epsilon(e_r\otimes e_s)&=\delta_{rs},
&e&=1_R.
\end{aligned}
\]

\begin{definition}[Green spiders]
Let
\[
t=(t_0,t_1,t_2,t_3)=(1,3+2\delta,3+3\delta,4),
\qquad
t^{-1}=(1,2+2\delta,2+3\delta,4).
\]
For each $k\in\Z$, define
\[
Z^k_{n,m}:V^{\otimes n}\longrightarrow V^{\otimes m}
\]
by
\[
Z^k_{n,m}
=\sum_{r=0}^{3}t_r^k
\ket{e_r^{\otimes m}}\bra{e_r^{\otimes n}},
\qquad e_r^{\otimes0}=1.
\]
Equivalently, on basis elements,
\[
Z^k_{n,m}(e_{a_1}\otimes\cdots\otimes e_{a_n})
=
\begin{cases}
t_r^k e_r^{\otimes m},&a_1=\cdots=a_n=r,\\
0,&\text{otherwise}.
\end{cases}
\]
Since each $t_r$ is a unit of $R$, $t_r^k$ is defined for all $k\in\Z$.
\end{definition}

The interpretation below is an interpretation of the integer-labelled syntax of Backens--Perdrix--Wang \cite{backens2020towards} described in Section 2, so it does not factor through the quotient identifying \(k\sim k+4\). Indeed,
\[
(3+2\delta)^4 = 1+\delta \neq 1,
\qquad
(3+3\delta)^4 = 1+4\delta \neq 1.
\]
This convention is compatible with our interpretation, since \(\ZXsimp\setminus\{(B2')\}\) is a ruleset in which \(2\pi\)-periodicity is not assumed syntactically.

\begin{definition}[Hadamard]
Let
\[
H_0=
\begin{pmatrix}
1&1\\
1&4
\end{pmatrix}^{\!\otimes2}
=
\begin{pmatrix}
1&1&1&1\\
1&4&1&4\\
1&1&4&4\\
1&4&4&1
\end{pmatrix},
\qquad
H_1=
\begin{pmatrix}
0&0&0&0\\
0&4&0&1\\
0&0&1&4\\
0&1&4&0
\end{pmatrix},
\]
and define
\[
H:=\frac12\left(H_0+\delta H_1\right).
\]
By direct computation,
\[
H^T=H,
\qquad
H^2=I.
\]
\end{definition}

\begin{definition}[Red spiders]
For each $k\in\Z$, define
\[
X^k_{n,m}:=H^{\otimes m}\circ Z^k_{n,m}\circ H^{\otimes n}.
\]
\end{definition}

\begin{definition}[Interpretation functor]
For any ZX diagram $D:n\to m$, define
\[
\llbracket D\rrbracket^{(B2')}:V^{\otimes n}\longrightarrow V^{\otimes m}
\]
inductively by
\[
\llbracket D_1\otimes D_2\rrbracket^{(B2')}
=\llbracket D_1\rrbracket^{(B2')}\otimes\llbracket D_2\rrbracket^{(B2')},
\qquad
\llbracket D_2\circ D_1\rrbracket^{(B2')}
=\llbracket D_2\rrbracket^{(B2')}\circ\llbracket D_1\rrbracket^{(B2')},
\]
and on generators by
\[
\begin{gathered}
\semzx{generator-identity}=I
\qquad
\semzx{generator-braiding}=\sigma
\\[2.1ex]
\semzx{generator-cup}=\eta
\qquad
\semzx{generator-cap}=\epsilon
\\[4.2ex]
  \intf{{\let\alpha k\vcenter{\hbox{\begin{tikzpicture}\footnotesize
	\begin{pgfonlayer}{nodelayer}
		\node [style=none] (0) at (-0.5, 0.5) {};
		\node [style=none] (1) at (-1, 0.5) {};
		\node [style=none] (2) at (0, 0.45) {...};
		\node [style=none] (3) at (0, -0.45) {...};
		\node [style=none] (4) at (0.5, -0.5) {};
		\node [style=none] (5) at (0.75, 0.75) {};
		\node [style=none] (6) at (0.75, -0.75) {};
		\node [style=gn] (7) at (-0.25, -0) {$~\alpha~$};
		\node [style=none] (8) at (-1.25, 0.75) {};
		\node [style=none] (9) at (-1.25, -0.75) {};
		\node [style=none] (10) at (-1, -0.5) {};
		\node [style=none] (11) at (0.5, 0.5) {};
		\node [style=none] (12) at (-0.5, -0.5) {};
	\end{pgfonlayer}
	\begin{pgfonlayer}{edgelayer}
		\draw [style=braceedge] (6.center) to node[wire label, inner sep={5 pt}]{$m$} (9.center);
		\draw [style=none, bend left=15, looseness=1.00] (10.center) to (7);
		\draw [style=none, bend left=15, looseness=1.00] (12.center) to (7);
		\draw [style=braceedge] (8.center) to node[wire label, inner sep={5 pt}]{$n$} (5.center);
		\draw [style=none, bend left=15, looseness=1.00] (7) to (1.center);
		\draw [style=none, bend right=15, looseness=1.00] (4.center) to (7);
		\draw [style=none, bend right=15, looseness=1.00] (7) to (11.center);
		\draw [style=none, bend left=15, looseness=1.00] (7) to (0.center);
	\end{pgfonlayer}
\end{tikzpicture}}}}}^{(B2')}=Z^k_{n,m}
  \qquad
  \intf{{\let\alpha k\vcenter{\hbox{\begin{tikzpicture}\footnotesize
	\begin{pgfonlayer}{nodelayer}
		\node [style=none] (0) at (1.25, -0.5) {};
		\node [style=none] (1) at (2.25, 0.5) {};
		\node [style=none] (2) at (0.5, -0.75) {};
		\node [style=none] (3) at (0.75, 0.5) {};
		\node [style=none] (4) at (0.5, 0.75) {};
		\node [style=rn] (5) at (1.5, 0) {$~\alpha~$};
		\node [style=none] (6) at (2.5, -0.75) {};
		\node [style=none] (7) at (1.75, -0.45) {...};
		\node [style=none] (8) at (0.75, -0.5) {};
		\node [style=none] (9) at (2.25, -0.5) {};
		\node [style=none] (10) at (1.25, 0.5) {};
		\node [style=none] (11) at (1.75, 0.45) {...};
		\node [style=none] (12) at (2.5, 0.75) {};
	\end{pgfonlayer}
	\begin{pgfonlayer}{edgelayer}
		\draw [style=braceedge] (6.center) to node[wire label, inner sep={5 pt}]{$m$} (2.center);
		\draw [style=none, bend left=15, looseness=1.00] (8.center) to (5);
		\draw [style=none, bend left=15, looseness=1.00] (0.center) to (5);
		\draw [style=braceedge] (4.center) to node[wire label, inner sep={5 pt}]{$n$} (12.center);
		\draw [style=none, bend left=15, looseness=1.00] (5) to (3.center);
		\draw [style=none, bend right=15, looseness=1.00] (9.center) to (5);
		\draw [style=none, bend right=15, looseness=1.00] (5) to (1.center);
		\draw [style=none, bend left=15, looseness=1.00] (5) to (10.center);
	\end{pgfonlayer}
\end{tikzpicture}}}}}^{(B2')}=X^k_{n,m}
  \\[4.2ex]
\semzx{generator-hadamard}=H
\qquad
\semzx{generator-unit}=e.
\end{gathered}
\]
\end{definition}

\begin{lemma}[Connectivity]
The interpretation $\llbracket-\rrbracket^{(B2')}$ respects the connectivity meta-rule of Backens--Perdrix--Wang \cite{backens2020towards} ``only connectivity matters.''
\end{lemma}

\begin{proof}
The structural maps are the usual compact structure for the basis $(e_0,e_1,e_2,e_3)$:
\[
\eta(1)=\sum_{r=0}^{3}e_r\otimes e_r,
\qquad
\epsilon(e_r\otimes e_s)=\delta_{rs},
\]
so the yanking equations hold:
\[
(\epsilon\otimes I)(I\otimes\eta)(e_a)=e_a,
\qquad
(I\otimes\epsilon)(\eta\otimes I)(e_a)=e_a.
\]
For green spiders,
\[
Z^k_{n,m}=\sum_{r=0}^{3}t_r^k
\ket{e_r^{\otimes m}}\bra{e_r^{\otimes n}}.
\]
This formula is unchanged by permuting input legs or output legs. It is also unchanged by bending wires, since the cup and cap identify $e_r$ with the dual basis vector $e_r^*$.

The Hadamard satisfies
\[
H^T=H,
\qquad
H^2=I.
\]
Hence Hadamard boxes can also be bent through the compact structure. Finally, red spiders are defined by
\[
X^k_{n,m}=H^{\otimes m}Z^k_{n,m}H^{\otimes n},
\]
so the same permutation and bending properties pass from green spiders to red spiders.

Therefore the interpreted map depends only on the labelled multigraph of the diagram, with the external input and output labels fixed. This is exactly the connectivity meta-rule in \cite{backens2020towards}.
\end{proof}

\begin{theorem}
\label{thm:b2-necessary}
The $(B2')$ rule is necessary: $\ZXsimp\setminus\{(B2')\}\nvdash(B2')$.
\end{theorem}

\begin{proof}
We evaluate $\llbracket-\rrbracket^{(B2')}$ on each rule to find agreement on all $\ZXsimp\setminus\{(B2')\}$ rules (see \cref{app:proof-theorem-5-6} for details).

It is enough to find one basis vector and one coefficient on which the two sides disagree.

However, in the case of $(B2')$, evaluating both sides on $e_1\otimes e_1$ and comparing the $e_0\otimes e_2$ coefficient\ddaggerfootnote{As before, see the \hyperref[app:bialgebra-matrices]{appendix} for the full \(16\times16\) matrices, where we have highlighted the \((e_1\otimes e_1,e_0\otimes e_2)\) coordinate.} gives
\[
\semzx{b2prime-lhs}
\text{ has $e_0\otimes e_2$ coefficient $\delta$,}\quad
\text{and}\quad
\semzx{b2prime-rhs}
\text{ has $e_0\otimes e_2$ coefficient $0$.}
\]
But, $\delta\neq0$ in $R$, so the two sides do not agree. Therefore, $\mathrm{LHS}\neq\mathrm{RHS}$ and $(B2')$ fails.

\[
\text{Thus, }\ZXsimp\setminus\{(B2')\}\nvdash(B2').
\]
\end{proof}

The independence of $(B2')$ is particularly notable, since this rule expresses the bialgebra interaction between the two observables, and hence encodes a core complementarity principle of the calculus.

\begin{corollary}
\label{cor:minimality}
Every rule in $\ZXsimp$ is necessary relative to the connectivity meta-rule of Backens--Perdrix--Wang \cite{backens2020towards}. Hence, relative to the meta-rule, $\ZXsimp$ has no redundant rewrite rule.
\end{corollary}

\begin{proof}
Immediate by \cite{backens2020towards} and the two theorems above.
\end{proof}

\section*{Acknowledgements}

I thank Miriam Backens and Simon Perdrix for helpful discussions about this work and for checking the computations in an earlier version of this paper.

\section*{Author Contributions}

The author is the sole author of this work and carried out all scientific aspects of the project, including the mathematical arguments, calculations, writing, and revision of the manuscript. Large language model tools were used to assist with formatting and grammar. The author is solely responsible for the work.

\bibliographystyle{quantum}
\bibliography{references}

\appendix

\section{Appendix}
\label{sec:appendix}

\subsection{Proof of Lemma 3.2}
\label{app:proof-lemma-3-2}

\begin{proof}
Let \(D\) be a ZX diagram.

We write \(\intf{g}^{(\mathrm{S3}'\mathrm{R})}=c(g)\cdot\intf{g}\) for each generator \(g\).
By definition of \(\intf{-}^{(\mathrm{S3}'\mathrm{R})}\),
\[
c(Z^\alpha_{n,m})=i^{n+m-2},\quad
c(X^\beta_{n,m})=-1,\quad
c(H)=i,
\]
and
\[
c(I)=c(\sigma)=c(\eta)=c(\epsilon)=c(e)=1.
\]
If \(D=D_2\circ D_1\), then
\[
\begin{aligned}
\intf{D}^{(\mathrm{S3}'\mathrm{R})}
&=\intf{D_2\circ D_1}^{(\mathrm{S3}'\mathrm{R})}
=\intf{D_2}^{(\mathrm{S3}'\mathrm{R})}\intf{D_1}^{(\mathrm{S3}'\mathrm{R})}\\
&=c(D_2)c(D_1)\intf{D_2}\circ\intf{D_1}\\
&=c(D_2)c(D_1)\intf{D}.
\end{aligned}
\]
If \(D=D_1\otimes D_2\), then
\[
\begin{aligned}
\intf{D}^{(\mathrm{S3}'\mathrm{R})}
&=\intf{D_1\otimes D_2}^{(\mathrm{S3}'\mathrm{R})}
=\intf{D_1}^{(\mathrm{S3}'\mathrm{R})}\otimes\intf{D_2}^{(\mathrm{S3}'\mathrm{R})}\\
&=c(D_1)\cdot c(D_2)\cdot\intf{D_1}\otimes\intf{D_2}\\
&=c(D_1)c(D_2)\intf{D}.
\end{aligned}
\]
So, the total scalar for \(D\) is the product of the local scalars of the non-structural vertices of \(D\).

That is, we can simply count \& multiply together the factors from our nontrivial generators.

Each green spider \(v=Z^\alpha_{n,m}\) contributes \(i^{n+m-2}=i^{\deg(v)-2}\).

As a collection they contribute:
\[
\prod_{v\in Z(D)} i^{\deg(v)-2}
=i^{\sum_{v\in Z(D)}(\deg(v)-2)}.
\]
Each red spider contributes a factor of \(-1\), so as a collection they contribute:
\[
(-1)^{|X(D)|}.
\]
Each Hadamard contributes \(i\), so as a collection they contribute:
\[
(i)^{|H(D)|}.
\]
Combining these results, we get:
\[
c(D)=i^{\sum_{v\in Z(D)}(\deg(v)-2)}\cdot(-1)^{|X(D)|}\cdot(i)^{|H(D)|}.
\]
\end{proof}

\subsection{Proof of Theorem 5.6}
\label{app:proof-theorem-5-6}

\emph{Note:} Maps \(0\to m\) are determined by their value on \(1\). In the other cases, we decide equality by evaluating on the relevant basis vectors or by direct matrix computation.

\begin{proof}
\leavevmode\par
\begin{itemize}[leftmargin=*,itemsep=1.25em]
\item \textbf{(S1).}
By definition 5.1, \(Z_{n,m}^k\) is the basis-copying spider with pointwise coefficient \(t_a^k\).

So, we can apply the generalized spider theorem for Frobenius algebras \cite[Theorem~4.11]{coecke2011interacting}.

Therefore,
\[
\begin{aligned}
\intf{\sOneLhsDiagram}^{(\mathrm{B2}')}
&=(\mathrm{id}_q\otimes Z_{r+1,s}^\ell)\circ
\big(Z_{p,q+1}^k\otimes \mathrm{id}_r\big)\\
&=Z_{p+r,q+s}^{k+\ell}
=\intf{\sOneRhsDiagram}^{(\mathrm{B2}')}.
\end{aligned}
\]
Concretely, the connecting wire forces a common basis label \(a\), \& the scalar on that label is
\[
t_a^k t_a^\ell=t_a^{k+\ell},
\]
while all non-matching labels vanish on both sides.

\item \textbf{(S3'L).}
Evaluating on \((1)\) on both sides,
\[
\intf{\cupDiagram}^{(\mathrm{B2}')}(1)
=\eta(1)=\sum_{r=0}^3 e_r\otimes e_r,
\]
and
\[
\intf{\greenCupDiagram}^{(\mathrm{B2}')}(1)
=Z_{0,2}^0(1)
=\sum_{r=0}^3 t_r^0 e_r^{\otimes 2}
=\sum_{r=0}^3 e_r\otimes e_r.
\]

\item \textbf{(S3'R).}
Evaluating on \((1)\) on both sides,

Note that \((H)^T=H\) \& \((H)^2=I\).
\[
\begin{aligned}
\intf{\redCupDiagram}^{(\mathrm{B2}')}(1)
&=X_{0,2}^0(1)\\
&=(H\otimes H)Z_{0,2}^0(1)\\
&=(H\otimes H)\sum_{r=0}^3 t_r^0 e_r^{\otimes 2}\\
&=(H\otimes H)\eta(1)\\
&=(I\otimes H(H)^T)\eta(1)\\
&=(I\otimes I)\eta(1)=\eta(1)\\
&=\intf{\cupDiagram}^{(\mathrm{B2}')}(1).
\end{aligned}
\]

\item \textbf{(B1).}
We use as shorthand:
\[
\rho:=X_{0,1}^0,\qquad
u:=Z_{0,1}^0=\sum_{r=0}^3 e_r,\qquad
\bar u:=Z_{1,0}^0,\qquad
\Delta:=Z_{1,2}^0.
\]
Note that the row sums of \(2H\) are \((4,0,0,0)\), so \((2H)\circ u=(4,0,0,0)\).

Now,
\[
\rho=H\circ u=\frac{1}{2}(2H)\circ u
=\frac{1}{2}(4,0,0,0)=2\cdot e_0.
\]
Then,
\[
(\bar u\circ\rho)=2(\bar u\cdot e_0)=2.
\]
And,
\[
(\Delta\rho)=2(\Delta e_0)=2(e_0\otimes e_0).
\]
Therefore,
\[
\intf{\bOneLhsDiagram}^{(\mathrm{B2}')}=4(e_0\otimes e_0),
\]
\[
\intf{\bOneRhsDiagram}^{(\mathrm{B2}')}=\rho\otimes\rho=4(e_0\otimes e_0).
\]

\item \textbf{(ZO').}
\[
Z_{0,0}^2=0,
\]
so
\[
\intf{\zOLhsDiagram}^{(\mathrm{B2}')}=0=\intf{\zORhsDiagram}^{(\mathrm{B2}')}.
\]

\item \textbf{(H).}
Note that \((H)^2=I\).
\[
\begin{aligned}
\intf{\hLhsDiagram}^{(\mathrm{B2}')}
&=(H)^{\otimes m}\circ X_{n,m}^k\circ (H)^{\otimes n}\\
&=(H)^{\otimes m}\circ (H)^{\otimes m}\circ
Z_{n,m}^k\circ (H)^{\otimes n}\circ (H)^{\otimes n}\\
&=Z_{n,m}^k
=\intf{\hRhsDiagram}^{(\mathrm{B2}')}.
\end{aligned}
\]

\item \textbf{(IV').}
Consider each ``dumbbell''
\[
\chi:=Z_{3,0}^0\circ X_{0,3}^0.
\]
Let \(H=(h_{ab})_{a,b}\).
Then
\[
\chi =
\sum_{b=0}^3\sum_{a=0}^3 h_{ab}^3
=8+(10+5\delta)+(10+5\delta)+(10+5\delta)
=3
\]
in \(R\).
Also,
\[
Z_{0,0}^0=\sum_{r=0}^3 t_r^0=1+1+1+1=4.
\]
So,
\[
\intf{\ivLhsDiagram}^{(\mathrm{B2}')}=4\cdot 3\cdot 3=36=1,
\]
\[
\intf{\ivRhsDiagram}^{(\mathrm{B2}')}=1.
\]

\item \textbf{(B2').}
Now, let
\[
F:=X^0_{2,1}
=
H\circ Z^0_{2,1}\circ(H\otimes H).
\]
We compute \(F(e_1\otimes e_1)\). First,
\[
He_1 = 3e_0 + (2+2\delta)e_1 +3e_2 + (2+3\delta)e_3.
\]
Thus,
\[
Z^0_{2,1}(He_1\otimes He_1)
=
3^2e_0+(2+2\delta)^2e_1+3^2e_2+(2+3\delta)^2e_3
\]
\[
=
4e_0+(4+3\delta)e_1+4e_2+(4+2\delta)e_3.
\]
Applying \(H\), we get
\[
F(e_1\otimes e_1)
=3e_0+3\delta e_2+2\delta e_3.
\]
(note that \(\sigma\) does nothing to \(e_1\otimes e_1\)).
\[
\begin{aligned}
\intf{\bTwoLhsDiagram}^{(\mathrm{B2}')}(e_1\otimes e_1)
&=Z_{0,0}^0\otimes F(e_1\otimes e_1)\otimes F(e_1\otimes e_1)\\
&=4(F(e_1\otimes e_1)\otimes F(e_1\otimes e_1))\\
&=4\big((3e_0+3\delta e_2+2\delta e_3)\otimes(3e_0+3\delta e_2+2\delta e_3)\big)
\end{aligned}
\]
Comparing the \(e_0\otimes e_2\) coefficient, we get
\[
4\cdot 3\cdot 3\delta = 36\delta=\delta
\quad\text{in } R.
\]
Now,
\[
X_{0,1}^0(1)=H(e_0+e_1+e_2+e_3)=2e_0,
\]
so,
\[
Z_{1,0}^0\circ X_{0,1}^0(1)=2.
\]
Thus,
\[
\begin{aligned}
\intf{\bTwoRhsDiagram}^{(\mathrm{B2}')}(e_1\otimes e_1)
&=2\cdot Z_{1,2}^0\circ F(e_1\otimes e_1)\\
&=2Z_{1,2}^0(3e_0+3\delta e_2+2\delta e_3)\\
&=e_0\otimes e_0+\delta e_2\otimes e_2+4\delta e_3\otimes e_3.
\end{aligned}
\]
Thus the \(e_0\otimes e_2\) coefficient on the right-hand side is \(0\).
But, \(\delta \neq 0\) in \(R\), so the two sides do not agree.
\[
\therefore\quad (\mathrm{B2}')\ \text{fails under}\ \intf{-}^{(\mathrm{B2}')}.
\]

\item \textbf{(EU').}
Note that \(Z_{2,1}^0\) just multiplies matching coordinates.

For vectors \(u,v\in V\):
\[
Z_{2,1}^0(u\otimes v)=\sum_{r=0}^3 u_rv_r e_r.
\]
Define \(D=\diag(t_0,t_1,t_2,t_3)\), the diagonal matrix of our phases:
\[
t=
\begin{pmatrix}
t_0\\ t_1\\ t_2\\ t_3
\end{pmatrix}
=
\begin{pmatrix}
1\\ 3+2\delta\\ 3+3\delta\\ 4
\end{pmatrix},
\qquad
t^{-1}=
\begin{pmatrix}
1\\ 2+2\delta\\ 2+3\delta\\ 4
\end{pmatrix}.
\]
\[
Ht^{-1}=
\begin{pmatrix}
2\\ 1+4\delta\\ 1+\delta\\ 3
\end{pmatrix}
=2
\begin{pmatrix}
1\\ 3+2\delta\\ 3+3\delta\\ 4
\end{pmatrix}
=2t
\quad\Longrightarrow\quad
\diag(Ht^{-1})=2D.
\]
Then,
\[
Z_{1,1}^1=D
\quad\&\quad
Z_{0,1}^{-1}=t^{-1}.
\]
So,
\[
\begin{aligned}
\intf{\euRhsDiagram}^{(\mathrm{B2}')}
&=Z_{1,1}^1\circ X_{2,1}^0\circ
\big(Z_{1,1}^1\otimes Z_{0,1}^{-1}\big)\\
&=D\circ H\circ Z_{2,1}^0\circ (H\otimes H)\circ(D\otimes t^{-1})\\
&=D\circ H\circ Z_{2,1}^0\circ(HD\otimes Ht^{-1})\\
&=D\circ H\circ \diag(Ht^{-1})\circ H\circ D\\
&=2(DHDHD)=18(D(2H)D(2H)D)=3(2H)=H,
\end{aligned}
\]
where \(D(2H)D(2H)D=2H\) by direct computation and \(18=3\) in \(R\).
\[
\therefore\quad
\intf{\euLhsDiagram}^{(\mathrm{B2}')}
=\intf{\euRhsDiagram}^{(\mathrm{B2}')}.
\]
\end{itemize}
\end{proof}

\clearpage
\subsection{The Bialgebra Matrices}
\label{app:bialgebra-matrices}

Here we present the \(16\times16\) matrices for both sides of the bialgebra interpretation, referenced in the proof of \cref{thm:b2-necessary}, and the derivation of our bialgebra interpretation from \cref{sec:bialgebra-rule}.

\begingroup
\setcounter{MaxMatrixCols}{16}
\newcommand{\matzero}{\textcolor{black!25}{0}}
\newcommand{\redbox}[1]{{\setlength{\fboxsep}{1pt}\color{red}\boxed{\color{black}#1}}}
\begin{center}
\resizebox{\linewidth}{!}{$
M_{\mathrm{LHS}}=
\begin{pmatrix}
1&\matzero&\matzero&\matzero&\matzero&1&\matzero&\matzero&\matzero&\matzero&1&\matzero&\matzero&\matzero&\matzero&1\\
\matzero&\matzero&\matzero&\matzero&\matzero&\matzero&\matzero&\matzero&\matzero&\matzero&4\lambda\delta&\matzero&\matzero&\matzero&\matzero&\lambda\delta\\
\matzero&\matzero&\matzero&\matzero&\matzero&\redbox{\lambda\delta}&\matzero&\matzero&\matzero&\matzero&\matzero&\matzero&\matzero&\matzero&\matzero&4\lambda\delta\\
\matzero&\matzero&\matzero&\matzero&\matzero&4\lambda\delta&\matzero&\matzero&\matzero&\matzero&\lambda\delta&\matzero&\matzero&\matzero&\matzero&\matzero\\
\matzero&\matzero&\matzero&\matzero&\matzero&\matzero&\matzero&\matzero&\matzero&\matzero&4\lambda\delta&\matzero&\matzero&\matzero&\matzero&\lambda\delta\\
\matzero&1&\matzero&\matzero&1&\matzero&\matzero&\matzero&\matzero&\matzero&\matzero&1&\matzero&\matzero&1&\matzero\\
\matzero&\matzero&\matzero&\matzero&\matzero&\matzero&\matzero&4\lambda\delta&\matzero&\matzero&\matzero&\lambda\delta&\matzero&4\lambda\delta&\lambda\delta&\matzero\\
\matzero&\matzero&\matzero&\matzero&\matzero&\matzero&\lambda\delta&\matzero&\matzero&\lambda\delta&\matzero&4\lambda\delta&\matzero&\matzero&4\lambda\delta&\matzero\\
\matzero&\matzero&\matzero&\matzero&\matzero&\lambda\delta&\matzero&\matzero&\matzero&\matzero&\matzero&\matzero&\matzero&\matzero&\matzero&4\lambda\delta\\
\matzero&\matzero&\matzero&\matzero&\matzero&\matzero&\matzero&4\lambda\delta&\matzero&\matzero&\matzero&\lambda\delta&\matzero&4\lambda\delta&\lambda\delta&\matzero\\
\matzero&\matzero&1&\matzero&\matzero&\matzero&\matzero&1&1&\matzero&\matzero&\matzero&\matzero&1&\matzero&\matzero\\
\matzero&\matzero&\matzero&\matzero&\matzero&\matzero&4\lambda\delta&\lambda\delta&\matzero&4\lambda\delta&\matzero&\matzero&\matzero&\lambda\delta&\matzero&\matzero\\
\matzero&\matzero&\matzero&\matzero&\matzero&4\lambda\delta&\matzero&\matzero&\matzero&\matzero&\lambda\delta&\matzero&\matzero&\matzero&\matzero&\matzero\\
\matzero&\matzero&\matzero&\matzero&\matzero&\matzero&\lambda\delta&\matzero&\matzero&\lambda\delta&\matzero&4\lambda\delta&\matzero&\matzero&4\lambda\delta&\matzero\\
\matzero&\matzero&\matzero&\matzero&\matzero&\matzero&4\lambda\delta&\lambda\delta&\matzero&4\lambda\delta&\matzero&\matzero&\matzero&\lambda\delta&\matzero&\matzero\\
\matzero&\matzero&\matzero&1&\matzero&\matzero&1&\matzero&\matzero&1&\matzero&\matzero&1&\matzero&\matzero&\matzero
\end{pmatrix}.
$}

\medskip

\resizebox{\linewidth}{!}{$
M_{\mathrm{RHS}}=
\begin{pmatrix}
1&\matzero&\matzero&\matzero&\matzero&1&\matzero&\matzero&\matzero&\matzero&1&\matzero&\matzero&\matzero&\matzero&1\\
\matzero&\matzero&\matzero&\matzero&\matzero&\matzero&\matzero&\matzero&\matzero&\matzero&\matzero&\matzero&\matzero&\matzero&\matzero&\matzero\\
\matzero&\matzero&\matzero&\matzero&\matzero&\redbox{\matzero}&\matzero&\matzero&\matzero&\matzero&\matzero&\matzero&\matzero&\matzero&\matzero&\matzero\\
\matzero&\matzero&\matzero&\matzero&\matzero&\matzero&\matzero&\matzero&\matzero&\matzero&\matzero&\matzero&\matzero&\matzero&\matzero&\matzero\\
\matzero&\matzero&\matzero&\matzero&\matzero&\matzero&\matzero&\matzero&\matzero&\matzero&\matzero&\matzero&\matzero&\matzero&\matzero&\matzero\\
\matzero&1&\matzero&\matzero&1&\matzero&\lambda\delta&4\lambda\delta&\matzero&\lambda\delta&4\lambda\delta&1&\matzero&4\lambda\delta&1&\lambda\delta\\
\matzero&\matzero&\matzero&\matzero&\matzero&\matzero&\matzero&\matzero&\matzero&\matzero&\matzero&\matzero&\matzero&\matzero&\matzero&\matzero\\
\matzero&\matzero&\matzero&\matzero&\matzero&\matzero&\matzero&\matzero&\matzero&\matzero&\matzero&\matzero&\matzero&\matzero&\matzero&\matzero\\
\matzero&\matzero&\matzero&\matzero&\matzero&\matzero&\matzero&\matzero&\matzero&\matzero&\matzero&\matzero&\matzero&\matzero&\matzero&\matzero\\
\matzero&\matzero&\matzero&\matzero&\matzero&\matzero&\matzero&\matzero&\matzero&\matzero&\matzero&\matzero&\matzero&\matzero&\matzero&\matzero\\
\matzero&\matzero&1&\matzero&\matzero&\lambda\delta&4\lambda\delta&1&1&4\lambda\delta&\matzero&\lambda\delta&\matzero&1&\lambda\delta&4\lambda\delta\\
\matzero&\matzero&\matzero&\matzero&\matzero&\matzero&\matzero&\matzero&\matzero&\matzero&\matzero&\matzero&\matzero&\matzero&\matzero&\matzero\\
\matzero&\matzero&\matzero&\matzero&\matzero&\matzero&\matzero&\matzero&\matzero&\matzero&\matzero&\matzero&\matzero&\matzero&\matzero&\matzero\\
\matzero&\matzero&\matzero&\matzero&\matzero&\matzero&\matzero&\matzero&\matzero&\matzero&\matzero&\matzero&\matzero&\matzero&\matzero&\matzero\\
\matzero&\matzero&\matzero&\matzero&\matzero&\matzero&\matzero&\matzero&\matzero&\matzero&\matzero&\matzero&\matzero&\matzero&\matzero&\matzero\\
\matzero&\matzero&\matzero&1&\matzero&4\lambda\delta&1&\lambda\delta&\matzero&1&\lambda\delta&4\lambda\delta&1&\lambda\delta&4\lambda\delta&\matzero
\end{pmatrix}.
$}
\end{center}
\endgroup

Taking \(\lambda=1\) gives the exact matrix obtained in the interpretation.
The \((e_1\otimes e_1,e_0\otimes e_2)\) coordinate is highlighted in red.

\end{document}